\documentclass[11pt]{article}
\usepackage{lscape,tikz}
\definecolor{darkblue}{rgb}{0.2,0.2,0.71}
\usepackage{todonotes}
\usepackage{epsfig,cancel,ulem, amsthm,amssymb}
\usepackage{color,soul}
\usetikzlibrary{decorations.markings,arrows}

\usepackage{wrapfig}
\usepackage{graphicx,framed}
\usepackage[colorlinks=true, pdfstartview=FitV, linkcolor=darkblue, citecolor=darkblue, urlcolor=darkblue]{hyperref}
\definecolor{shadecolor}{rgb}{0.95, 0.95, 0.86}
\definecolor{darkgreen}{rgb}{0.2, 0.5,  0}

\newcommand{\Ker}{\mathrm{Ker}}

\def \Im { {\rm Im}}

\def\&{\vspace{-5pt}&}
\def\Tr{ {\rm Tr}}

\def\C{\mathbb C}
\def\tr{{\rm tr}}
\def \HH {H}
\def \TT{\mathbf T}

\def\gl{{\rm gl}}
\def\rc { {\rho^\vee }}

\def \eqref#1{(\ref{#1})}
\def \& {&\hspace{-10pt}}
\def \Etheta{ E_\theta }
\def \Ftheta{ E_{-\theta}}

\def \wt{\widetilde}

\newcommand{\bt}{\beta}

\renewcommand{\d}{\mathrm d}

\newcommand{\pa}{\partial}
\newcommand{\p}{\partial}
\newcommand{\bdt}{{\bf t}}
\newcommand{\bdT}{{\bf T}}
\newcommand{\bdzero}{{\bf 0}}

\newcommand{\br}{{\mathbb R}}
 
\newcommand{\ERR}{\mathcal{R}}

\newcommand{\nn}{\nonumber}

\newcommand{\g}{\mathfrak{g}}
\newcommand{\fb}{\mathfrak{b}}
\newcommand{\fn}{\mathfrak{n}}
\newcommand{\s}{s}

\newtheorem{theorem}{Theorem}[subsection]
\newtheorem{example}[theorem]{Example}
\newtheorem{exercise}[theorem]{Exercise}

\newtheorem{lemma}[theorem]{Lemma}
\newtheorem{remark}[theorem]{Remark}

\newtheorem{proposition}[theorem]{Proposition} 
\newtheorem{corollary}[theorem]{Corollary} 
\newtheorem{definition}[theorem]{Definition}

\def\le{\left}
\def\ri{\right}
\def\ds{\displaystyle}
\def\V {\mathcal V}

\def\res{\mathop{\mathrm {res}}\limits_}

\def\bt{\begin{theorem}}
\def\et{\end{theorem}}
\def\bc{\begin{corollary}}
\def\ec{\end{corollary}}
\def\bx{\begin{example}\small}
\def\ex{\end{example}}
\def\bxr{\begin{exercise}\small}
\def\exr{\end{exercise}}
\def\bl{\begin{lemma}}
\def\el{\end{lemma}}
\def\bd{\begin{definition}}
\def\ed{\end{definition}}
\def\bp{\begin{proposition}}
\def\ep{\end{proposition}}

\def\br{\begin{remark}}
\def\er{\end{remark}}

\def\be{\begin{equation}}
\def\ee{\end{equation}}

\def\&{\hspace{-15pt}&}
\def\bea{\begin{eqnarray}}
\def\eea{\end{eqnarray}}

\def \pa{\partial}
\newcommand{\CC}{\mathbb{C}}
\def\L{\mathcal L}

\newcommand{\ZZ}{\mathbb{Z}}

\def\l{\lambda  }

\def\1{{\bf 1}}

\newcommand{\h}{\mathfrak{h}}

\def\QED {\hfill $\square$\par\vskip 10pt \noindent }

\newcommand{\ad}{\mathrm{ad}}

\makeatletter
\@addtoreset{equation}{section}
\makeatother

\begin{document}
\title{Simple Lie algebras, Drinfeld--Sokolov hierarchies, and multi-point correlation functions}
\author{Marco Bertola, Boris Dubrovin, Di Yang}
\date{}
\maketitle

\begin{abstract}
For a simple Lie algebra $\g$, we derive a simple algorithm for computing logarithmic derivatives of 
tau-functions of Drinfeld--Sokolov hierarchy of $\g$-type in terms of $\g$-valued resolvents.
We show, for the topological solution to the lowest-weight-gauge Drinfeld--Sokolov 
hierarchy of $\g$-type, the resolvents 
evaluated at zero satisfy the \textit{topological ODE}.
\end{abstract}


\tableofcontents

\setcounter{equation}{0}
\setcounter{theorem}{0}
\section{Introduction}\label{intro-s}
\setcounter{equation}{0}
\setcounter{theorem}{0}
\subsection{Simple Lie algebra and Drinfeld--Sokolov hierarchy}
Let~$\g$ be a simple Lie algebra over $\mathbb{C}$ of rank $n$, with the Lie bracket 
denoted by $[\cdot,\cdot].$ Let $\ad:\,\g\rightarrow \gl(\g)$ be the adjoint representation of $\g$. 
We denote by $h,\,h^\vee$ the Coxeter and dual Coxeter numbers~\cite{Kac} of~$\g$, and 
$m_1=1<m_2\leq \cdots\leq m_{n-1}<m_n=h-1$ the exponents.
Denote $(\cdot|\cdot): \g \times\g\rightarrow \mathbb{C}$ the {\it  normalized} Cartan--Killing \cite{Cartan} form
\be\label{norm-ck-fin}
(x|y):=\frac{1}{2h^\vee} {\tr\bigl(\ad_x \ad_y\bigr)},\qquad \forall\,x,y\in \g.
\ee

Fix a Cartan subalgebra  $\h \subset \g$,  and let $\triangle\subset \h^*$ be the root system. We choose a set of simple roots $\Pi=\{\alpha_1,\dots,\alpha_n\}\subset\h^*$.  
Then $\g$ has the root space decomposition
$$
\g=\h\oplus \bigoplus_{\alpha\in\triangle} \g_{\alpha}.
$$
For any $\alpha\in\triangle,$ denote by $H_{\alpha}$ the unique element in $\h$ such that $(H_\alpha | X) = \alpha(X), \ \ \forall  \, X\in \h$. 
The normalized Cartan--Killing form induces naturally a non-degenerate bilinear form on $\h^*:$
$$
(\alpha|\beta)=(H_{\alpha}| H_{\beta}),\quad \forall \, \alpha,\beta\in \h^*.
$$
Denote by $E_i\in\g_{\alpha_i},\,F_i\in\g_{-\alpha_i}$, $H_i={2H_{\alpha_i}}/{(\alpha_i|\alpha_i)}$ the Weyl generators of $\g$. They  satisfy
$$
[E_i,F_j]=H_i  \delta_{ij}, \quad [H_i,E_j]=A_{ij} E_j,\quad [H_i,F_j]=-A_{ij} F_j, 
$$
where $(A_{ij})$ denotes the Cartan matrix associated to $(\g, \Pi)$, and $\delta_{ij}$ is the Kronecker delta. Here and in what follows, free Latin indices take the integer values from~$1$ to~$n$ unless otherwise indicated.

Let $\theta$ be the highest root w.r.t. $\Pi$; recall that  $(\theta|\theta)=2.$ 
We choose  $ \Ftheta  \in \g_{-\theta},\, \Etheta \in \g_\theta,$ 
normalized by the conditions $(\Etheta  |  \Ftheta  ) = 1$ and $\omega( \Ftheta  ) = -\Etheta $, where $\omega : \g\rightarrow \g $ is the Chevalley involution. Let 
\be \label{def-I+}
I_+ := \sum_{i=1}^n  E_i
\ee
be a principal nilpotent element of $\g$. Define
\be \label{semi-DS}
\Lambda=I_+ + \lambda \Ftheta  .
\ee

Denote by $L(\g)=\g \otimes \CC[\lambda, \lambda^{-1}]$ the loop algebra of $\g$. 
The Lie bracket $[\cdot,\cdot]$ and the Cartan--Killing form $(\cdot|\cdot)$ extend naturally to $L(\g)$.
We have
\be
L(\g)= \Ker \, \ad_\Lambda \oplus \,  {\rm Im} \, \ad_\Lambda, \qquad \Ker \, \ad_\Lambda \perp \,  {\rm Im} \, \ad_\Lambda.
\ee
Recall that the {\it  principal gradation} on $L(\g)$ is defined by  
$$
\deg \lambda = h,\quad \deg E_i=-\deg F_i=1.
$$
Observe that
$$
\deg\Lambda=1.
$$
This gradation is of course also defined on $\g=\g\otimes 1$. 
With the principal gradation, the loop algebra $L(\g)$ and the simple Lie algebra $\g$ decompose
 into direct sums of homogeneous subspaces $L(\g)^j,\, \g^j, ~ j\in \mathbb{Z}:$
$$
L(\g)=\bigoplus_{j\in\mathbb{Z}} L(\g)^j,\qquad \g=\bigoplus_{j=-(h-1)}^{h-1} \, \g^j.
$$
We will denote the projection onto the nonnegative subspace by $(\bullet)^+: L(\g)\to \sum_{j\geq 0} L(\g)^j,$ and onto the negative subspace by $(\bullet)^-.$
\noindent It is known \cite{Kac1978} that $\Ker \, \ad_\Lambda \subset L(\g)$ admits the following decomposition
\bea
&& 
\Ker \, \ad_\Lambda=\bigoplus_{j \in E} \mathbb{C} \Lambda_j, \quad \Lambda_j\in L(\g)^j,\,j\in E, \nn \\
&& [\Lambda_i,  \Lambda_j ]=0,  \quad \forall \,i,j\in E. \nn
\eea
Here, $E:=\bigsqcup_{i=1}^n (m_i+h \mathbb{Z})$. 
The meaning of the symbol $\bigsqcup$ here  is that of ``disjoint union": this means that  if the exponents are distinct then~$\bigsqcup$ denotes the 
ordinary union, but if an element appears in more than one set, it is actually considered a new element. This is relevant  only for the case of the Lie algebra of type  $D_{n}$ with even $n=2k$: in this case  $m_{n/2+1}, m_{n/2+1}+h, \dots$ should be written as $m_{n/2+1}', (m_{n/2+1}+h)', \dots$ because, as integers, 
$m_{n/2+1}=m_{n/2}$.

We choose normalizations of~$\Lambda_j$, $j\in E$ satisfying
\bea
&& \Lambda_{m_a+kh}=\Lambda_{m_a}  \lambda^k,\quad k\in \ZZ, \label{norm-Lambda-1}\\
&& \bigl(\Lambda_{m_a}|\Lambda_{m_b}\bigr)=h \eta_{ab}\lambda. \label{norm-Lambda-2}
\eea
Here and below, 
\begin{equation}\label{eta}
\eta_{ab}: = \delta_{a+b, n+1}.
\end{equation}
Since $\Lambda\in L(\g)^1$, we fix the normalization of~$\Lambda_1$ such that
$$\Lambda_1=\Lambda.$$
It is useful to notice that $\Lambda_{m_a}$, $a=1,\dots,n$ have the form~\cite{Kostant}
$$
\Lambda_{m_a}=  L_{m_a} +  \l  \, K_{m_a-h} , \qquad L_{m_a} \in \g^{m_a} , ~ K_{m_a-h} \in \g^{m_a-h}, ~L_{m_a}\neq 0,\,K_{m_a-h}\neq 0. 
$$

In~\cite{DS}, Drinfeld and Sokolov associate to~$\g$ an integrable hierarchy of Hamiltonian evolutionary PDEs, 
known as the Drinfeld--Sokolov (DS) hierarchy of $\g$-type. Let us briefly review their construction in the form suitable for subsequent considerations. 
Denote by $\fb=\g^{\leq 0}$ a Borel subalgebra of~$\g$, and $\fn=\g^{<0}$ a nilpotent subalgebra. Let 
\be  \label{Lax}
\mathcal{L}=\p_x+ \Lambda + q(x),\quad q(x)\in \fb.
\ee

\bd\label{resol-defi}
The basic resolvents~$R_a$, $a=1,\dots,n$ of~$\L$ are defined as the unique solutions to
\bea
&\& [\L,R_a] \,=\, 0, \quad R_a \in  \mathcal{A}^q \otimes  \g((\lambda^{-1})), \label{basicdef1}\\
&\& R_a(\lambda;q,q_x,\cdots) \,=\, \Lambda_{m_a} + \mbox{ lower order terms w.r.t. } \deg , \label{basicdef2} \\
&\& \bigl(R_a(\lambda;q,q_x,\cdots)\, | \,R_b(\lambda;q,q_x,\cdots)\bigr) \,=\, h \, \eta_{ab}\, \lambda \label{boundaryr}
\eea
(here and below, $\mathcal{A}^q$ denotes the ring of differential polynomials in $q$, namely, an element of 
$\mathcal{A}^q$ is a polynomial 
in the entries of $q,q_x,q_{2x},\cdots$), together with the requirements that $R_a$ are homogeneous of degree $m_a$ with respect to the extended principal 
gradation defined by further assigning degrees to entries of~$q$ so that $q$ is homogeneous of degree~1.
\ed
\noindent Existence and uniqueness of the basic resolvents will be shown in Proposition~\ref{boundary}. 
Note that \eqref{boundaryr} can be alternatively replaced by the no-integration constant condition 
$R_a(\lambda;0,0,\cdots)=\Lambda_{m_a}$, which gives rise to a different algorithm of computing~$R_a$. 

The DS flows for the $\mathfrak b$-valued function $q=q(x,{\bf T})$, ${\bf T}=(T^a_k)^{a=1,\dots,n}_{k\geq 0}$ are 
an infinite set of compatible 
evolutionary PDEs of the form 	
\be\label{pre-DS}
\frac{\p \L}{\p T^a_k}= \le[\left(\lambda^k R_a\right)_+ \,,\, \L\ri], \quad k\geq 0.
\ee
The notation  $(\bullet)_+$ stands for the polynomial part of the expression  in the variable $\lambda$ 
(similarly, $(\bullet)_-$ will stand for Laurent tail in the variable ~$\lambda$). 
To see that these flows are well defined, we note that the property $[\L,R_a]=0$ implies that 
\be
\le[\left(\lambda^k R_a\right)_+ \,,\, \L\ri] = 
\p_x \le(\left(\lambda^k R_a\right)_-\ri) + \le[ \Lambda, \left(\lambda^k R_a\right)_-\ri] + \le[q,\left(\lambda^k R_a\right)_-\ri].
\ee
Then, observing that the RHS contains 
only non-positive powers in~$\lambda$ (here~\eqref{semi-DS} is used), and that the LHS contains 
only non-negative powers in~$\lambda$, we find that $\le[\left(\lambda^k R_a\right)_+ , \L\ri]$ takes value in $\g\otimes \lambda^0$.
Furthermore   this contribution can only come from the term 
$ \le[ \lambda E_{-\theta}, \left(\lambda^k R_a\right)_- \ri]$ (here \eqref{semi-DS} is used again): 
recalling that  $E_{-\theta}$ has the principal degree $-(h-1)$, we conclude that 
$\le[ \lambda E_{-\theta}, \left(\lambda^k R_a\right)_- \ri]\in \fb\otimes \lambda^0$.
An important property of these flows is that they pairwise commute \cite{DS}; they  
form the {\it pre-DS hierarchy}.

Consider transformations of the dependent variable $q(x)\mapsto \tilde q(x)$ of the pre-DS hierarchy induced by gauge transformations of the form
\be\label{gauge1}
{\mathcal L}=\p_x+\Lambda+q(x) \quad \mapsto \quad 
\widetilde{\mathcal L}= e^{\ad_{N(x)}} {\mathcal L}=\p_x+\Lambda+\tilde q(x)
\ee
for arbitrary $\fn$-valued smooth functions~$N(x)$.
A crucial point of the Drinfeld--Sokolov construction is the following statement.

\bl
The gauge transformations \eqref{gauge1} are symmetries of the pre-DS flows of \eqref{pre-DS}. In particular, they map solutions to solutions.
\el

In our approach the proof of this simple but important statement easily follows by observing
that the basic resolvents $\widetilde R_a$ of the gauge-transformed operator $\widetilde{\mathcal L}$ 
satisfy
\be\label{invres}
\widetilde R_a(\lambda; \tilde q, \tilde q_x, \cdots)=e^{\ad_{N(x)}}  R_a(\lambda; q, q_x, \cdots), \quad a=1, \dots, n.
\ee

The DS hierarchy is obtained from \eqref{pre-DS} by considering suitably chosen \textit{gauge invariant} functions $q^{\rm can}$ (see below for more details).

\setcounter{equation}{0}
\setcounter{theorem}{0}
\subsection{From resolvents to tau-function}

We start from defining tau-functions of an arbitrary solution $q(x,{\bf T})$ of the pre-DS hierarchy. Then we verify its independence from the choice of the gauge with respect to the transformations of the form \eqref{gauge1}.

\bd \label{tau-symm-def}
Define a sequence of functions $\Omega_{a,k;b,\ell}=\Omega_{a,k; b,\ell}(q, q_x, \cdots)\in{\mathcal A}^q$, $k,\ell,\geq 0$ by means of the generating function expression below
\be\label{tau-symm-eq}
\sum_{k,\ell\geq 0} \frac{\Omega_{a,k;b,\ell}}{\lambda^{k+1} \mu^{\ell+1}} =  \frac{ \le(R_a (\lambda) \, | \, R_b(\mu)\ri)}{(\lambda-\mu)^2} -
\eta_{ab} \frac{m_a  \lambda + m_b \mu}{(\lambda-\mu)^2}.
\ee
We call $\Omega_{a,k;b,\ell}$ the \textit{two-point correlation functions}.
\ed

\bl\label{tau-symmetry-def} 
The two-point correlation functions $\Omega_{a,k;b,\ell}$ satisfy the following properties
\bea
& \Omega_{a,k;b,\ell} \in \mathcal{A}^q,\qquad \Omega_{a,k;b,\ell}  =\Omega_{b,\ell;a,k},\quad \forall \, a,b,~ \forall\, k,\ell\geq 0, \label{symm1}\\
& \p_{T^c_m} \Omega_{a,k;b,\ell}   = \p_{T^a_k} \Omega_{b,\ell;c,m}= 
\p_{T^b_\ell} \Omega_{c,m;a,k}, \quad \forall \, a,b,c,~\forall\,k,\ell,m\geq 0. \label{symm2}
\eea
\el

\bl \label{tau-symm-def-xt}
For an arbitrary solution $q(x,\bdT)$ to \eqref{pre-DS},  
there exists $\tau=\tau(x,\bdT)$ such that
\bea 
&& \hspace{-2mm} \frac{\p^2 \log \tau}{\p T^a_k \p T^b_\ell} 
= \Omega_{a,k;b,\ell} \left(q(x,{\bf T}), q_x(x, {\bf T}), \cdots\right), \label{two-point-tau} \\
&&\hspace{-2mm}  \frac{\p \tau}{\p x} = -  \frac{\p \tau}{\p T^1_0}. \label{xt10}
\eea
\el
\noindent The proofs are provided later in the paper.

In view of \eqref{xt10} we will henceforth identify~$x$ with~$-T^1_0$ for~$\tau(x,\bdT)$. 
So we will use the short notation
$\tau=\tau(\bdT)$. Note that 
the scalar function~$\tau(\bdT)$ advocated for in Lemma~\ref{tau-symm-def-xt} is uniquely 
determined by the solution~$q(x,\bdT)$ only up to a factor of the form
\be\label{gauge-factor}
\exp\Biggl(d_0+\sum_{a=1}^n \sum_{k\geq 0} d_{a,k} T^a_k\Biggr),\qquad d_0,\,d_{a,k} \mbox{ arbitrary constants}.
\ee
\bd \label{our-def-tau}
We call~$\tau(\bdT)$ the {\rm tau-function} of the solution~$q(x,{\bf T})$ of the pre-DS hierarchy.
\ed

For related aspects on tau-functions, see for example \cite{BDY1} \cite{CW1} \cite{Dickey2} \cite{Du1} \cite{Du3} \cite{DZ-norm} \cite{FGM} 
\cite{G2} \cite{GM} \cite{HM} \cite{HMG} \cite{JMU} \cite{KW} \cite{Wu}.

\bd For an arbitrary solution to the pre-DS hierarchy, let~$\tau(\TT)$ be a tau-function of this solution in the sense of Definition \ref{our-def-tau}.  
The {\bf $N$-point correlation functions} of~$\tau(\bdT)$ are defined by
\bea\label{corf}
&\& \langle\langle\tau_{a_1 k_1} \cdots \tau_{a_N k_N}\rangle\rangle^{DS}=\frac{\p^N \log \tau}{\p T^{a_1}_{k_1}\dots \p 
T^{a_N}_{k_N}},
\quad k_1,\dots,k_N\geq 0,\,N\geq 1.
\eea
\ed

From \eqref{invres} it easily follows the following lemma.
\bl \label{tau-gauge-inv-lemma}
The tau-function of a solution to the pre-DS hierarchy is invariant, 
up to a factor of the form~\eqref{gauge-factor}, 
with respect to the gauge transformations~\eqref{gauge1}.
\el
\noindent Thus~$\tau(\bdT)$ will also be called tau-function of the solution~$q^{\rm can}$ of the DS hierarchy corresponding to a gauge-fixed Lax
operator. The usual procedure~\cite{DS} to fix the gauge is 
by choosing a subspace ${\mathcal V}\subset\mathfrak b$ transversal 
to the adjoint action of the nilpotent subgroup 
so that~$q^{\rm can}(x)$ restricts to a ${\mathcal V}$-valued function (see below).

\setcounter{equation}{0}
\setcounter{theorem}{0}
\subsection{Main results}

For any $a=1, \dots, n$ introduce the following differential operator depending on a parameter~$\lambda$
\be \label{nablaa}
\nabla_a(\lambda)=\sum_{k\geq 0}  \frac{\p_{T^a_k}}{\lambda^{k+1}}.
\ee

For a given $N\geq 1$ and a collection of integers $a_1,\dots,a_N\in \{1,\dots,n\}$, we define the following  
generating series of $N$-point correlations functions by
\be
F_{a_1,\dots,a_N}(\lambda_1,\dots,\lambda_N;\TT)
=\nabla_{a_1}(\lambda_1)\cdots \nabla_{a_N}(\lambda_N) \, \log\tau(\TT).
\ee
Observe that, for $N\geq 2$ the correlation functions \eqref{corf} depend only on the solution $q(x, {\bf T})$ of the pre-DS hierarchy.
Our goal is to derive an explicit expression for these generating functions for $N\geq 2$ in terms of the basic resolvents defined above.

For any $N\geq 2$ define a cyclic-symmetric $N$-linear form $B: \g\times\cdots \times\g\rightarrow \mathbb{C}$ by
\be
\label{Kn}
B(x_1,\dots,x_N)=\tr \le(\ad_{x_1} \dots  \ad_{x_N}\ri),\quad \forall\,x_1,\dots,x_N\in \g.
\ee
The normalized 
Cartan--Killing form (see~\eqref{norm-ck-fin}) and $B$ are related by $B(x,y)=2h^\vee(x|y)$.

\bt \label{main thm} 
For an arbitrary solution $q^{\rm can}(\bdT)$ to the DS-hierarchy,
let $\tau(\bdT)$ be a tau-function of this solution. 
Then $\forall\,N\geq 2$, we have
\bea
&& F_{a_1,\dots,a_N}(\lambda_1,\dots,\lambda_N;\bdT) = - \frac {1} {2 \, N \, h^\vee} \sum_{s \in S_N} 
\frac {B\le(R_{a_{s_1}}^{\rm can}(\lambda_{s_1}; \bdT),  \dots , R_{a_{s_N}}^{\rm can}(\lambda_{s_N};\bdT) \ri) }{\prod _{j=1}^{N} (\lambda_{s_j}- \lambda_{s_{j+1}})}\nn\\
&& \qquad \qquad \qquad\qquad\qquad\qquad\qquad\qquad 
-\delta_{N2} \, \eta_{a_1 a_2}  \frac {m_{a_1} \,  \lambda_1 + m_{a_2} \, \lambda_2}{(\lambda_1-\lambda_2)^2 },  \label{24}
\eea
where $R_a^{\rm can}(\lambda),\,a=1,\dots,n$ are the basic resolvents of $\L^{\rm can}:=\p_x+ \Lambda(\lambda)+ q^{\rm can}$, and it is understood that 
$s_{N+1} =s_1$.
In particular, $\forall \,N\geq 2, \, \forall\,a_1,\dots,a_N\in\{1,\dots,n\},$ we have
$F_{a_1,\dots,a_N}(\lambda_1,\dots,\lambda_N;\bdT)\in \mathcal{A}^{q^{\rm can}} [[\lambda_1^{-1}, \dots, \lambda_N^{-1}]].$
\et

\paragraph{The partition function.} We now  consider a particular
tau-function that we shall call the {\it partition function}: it will be denoted by~$Z(\bdt)$, 
where the new time variables~$\bdt$ differ from the original~${\bf T}$ by a rescaling (see eq.~\eqref{norm-tT}). 
This particular tau-function is uniquely specified up to a multiplicative constant by the following {\it string equation}:
\be\label{string}
\sum_{a=1}^n \sum_{k\geq 0} t^a_{k+1} \frac{\p Z}{\p t^a_k}+
\frac12 \sum_{a,b=1}^n \eta_{ab} t^{a}_0  t^b_0 Z = \frac{\p Z}{\p t^1_0}
\ee
(see details in Section~\ref{topo} below).
Here, the time variables $t^a_k$ and $T^a_k$ are related by
\be\label{norm-tT}
\frac{\partial}{\partial t^a_k}  =   {c_{a,k}}\frac{\partial}{\partial T^a_k},\quad c_{a,k}=\frac{ (-1)^{k}} { \sqrt{-h}^{m_a+hk+1} \, (\frac{m_a} h)_{k+1}}, \qquad  k\geq 0,
\ee
where $(\cdot)_\ell$ denotes the Pochhammer symbol, i.e.,
$(y)_\ell:=y(y+1)\cdots(y+\ell-1)$.

\bt \label{topo-thm} Let the subspace $\V:=\Ker\, \ad_{I_-}\subset \g$ be the
lowest weight gauge (see eq.~\eqref{def-I-} for the definition of $I_-$), and
$\L^{\rm can}$ the associated Lax operator.  Let $R^{\rm can}_a,\,a=1,\dots,n$ be the basic resolvents of $\L^{\rm can}.$ 
For the partition function $Z$,  define $M_a(\lambda)= \lambda^{-\frac{m_a} h} R^{\rm can}_a(\lambda;\bdt=\bdzero)$. Then $\forall\,a\in\{1,\dots,n\}$, 
$M_a(\lambda)$ satisfies the topological ODE of $\g$-type
\be\label{topoode}
M'= \kappa \, [M,\Lambda],\qquad \kappa=\le(\sqrt{-h}\ri)^{-h}, ~ {}' :=  \frac {\d}{\d \l}.
\ee
\et
See~\cite{BDY2} for the definition and properties of the topological ODE of $\g$-type.
Observe that, as $\lambda\to\infty,$ the solutions $M_a(\lambda)$ admit the expansions
$$
M_a=\lambda^{-\frac{m_a}h} \Bigl[ \Lambda_{m_a}+ \mbox{lower degree terms w.r.t.~} \deg\Bigr].
$$
Thus, $M_a$ coincide with the basis of regular solutions to the topological ODE constructed in~\cite{BDY2}.

\setcounter{equation}{0}
\setcounter{theorem}{0}
\subsection{Applications to the FJRW theory}  
Let $f:\mathbb{C}^m\rightarrow \CC$ be a quasi-homogeneous polynomial, i.e., there exist positive integers $d,n_1,\dots,n_m$, s.t.
$$f\le(z^{n_1}x_1,\dots,z^{n_m}x_m\ri)=z^d f(x_1,\dots,x_m), \quad \forall\, z\in \mathbb C.$$  
The weight of~$x_i$ is defined to be $q_i=\frac{n_i}d$, $i=1,\dots,m$.
In general the gradient of $f$ vanishes at the origin and hence 
the zero level-set $f^{-1}(0)$ is a singular variety and 
defines a ``singularity'' in the sense of singularity theory~\cite{AVG}. 
The function~$f$ is called \textit{non-degenerate} if the choice of weights~$q_i$ is unique 
and $x=\bdzero$ is the only singularity of~$f$. 
Let $G_f$ (or $G_{max}$) denote the maximal diagonal symmetry group of~$f$, 
which is the subgroup of~${\rm Aut}(f)$ consisting of diagonal matrices~$\gamma$ 
such that 
$f(\gamma x)=f(x)$. It is easy to see that the matrix 
$$J={\rm diag} \left(e^{2\pi i q_1},\dots,e^{2\pi i q_m}\right) \in G_f.$$
Let~$G$ be a subgroup of~$G_f$ containing $\langle J \rangle$. 
Let~$n$ 
be the dimension of the Fan--Jarvis--Ruan cohomology ring~\cite{FJR} associated to~$(f,G)$. 
Fan--Jarvis--Ruan associate with the pair $(f,G)$ a certain \textit{generalized Witten class}, 
called the Fan--Jarvis--Ruan--Witten class
$$\Lambda_{g,N}^{f,G}(a_1,\dots,a_N)\in H^*(\overline{\mathcal{M}}_{g,N}),\qquad a_1,\dots,a_N \in \{1,\dots,n\},$$
such that incorporation of these cohomological classes to 
$\overline{\mathcal{M}}_{g,N}$ gives rise to a cohomological filed theory \cite{Ma99, FJR} (cf.~also~\cite{Du1}, \cite{Du2}, \cite{Saito}). 
The FJRW invariants are defined by
$$\langle\tau_{a_1 k_1}\cdots\tau_{a_N k_N} \rangle_g^{f,G}=\int_{\overline{\mathcal{M}}_{g,N}} \psi_1^{k_1}\cdots\psi_N^{k_N} \, \Lambda_{g,N}^{f,G}(a_1,\dots,a_N),$$
where $\psi_i,\,i=1,\dots,N$ are $\psi$-classes. 
\bd
The partition function $Z^{f,G}$ of FJRW invariants is defined by
$$Z^{f,G}({\bf t})=\exp\le(\sum_{g,N\geq 0} \frac{1}{N!} \sum_{a_1,\dots,a_N=1}^n \sum_{k_1,\dots,k_N\geq 0}   \langle\tau_{a_1 k_1} \dots \tau_{a_N k_N}\rangle_g^{f,G}   \, t^{a_1}_{k_1}\cdots t^{a_N}_{k_N} \ri).$$
\ed
Now we consider an important subclass of singularities, called \textit{simple singularities}. 
They are classified by the ADE Dynkin diagrams \cite{A72,A76}. 
In particular, we consider
$$
A_k:~ f= x^{k+1},\quad k\geq 1; \qquad 
D_k:~ f=x^{k-1}+x\, y^2, \quad k\geq 4; 
$$
$$ 
E_6: ~ f=x^3+y^4; \qquad E_7: ~ f=x^3+x\,y^3; \qquad E_8:~f=x^3+y^5.
$$
We are also interested in the mirror singularity of $D_k$  \cite{FJR}, denoted by $D_k^T$:
$$D_k^T: \qquad f=x^{k-1}\,y+y^2,\quad k\geq 4.$$
The maximal diagonal symmetry groups $G_f$ of the above polynomials will be denoted by $G_{A_k}$, $G_{D_k}$, $G_{D_k^T}$ and $G_{E_n}$, $n=6, \, 7, \, 8$.

~~~

\noindent \textbf{Theorem-ADE} (\cite{FJR,FFJMR}). \textit{The following statements hold true
\begin{itemize}
\item[A.] The partition function $Z^{A_n,G}({\bf t}),\,n\geq 1$ with $G=\langle J \rangle=G_{A_n}$ is a particular tau-function of the Drinfeld--Sokolov hierarchy of $A_n$-type satisfying the string equation \eqref{string}.
\item[D.] The partition function $Z^{D_n,G}({\bf t}),\, n\geq 4$ with $n$ \textbf{even} and $G=\langle J \rangle$ is a particular tau-function of the DS hierarchy of $D_{n}$-type satisfying \eqref{string}.
\item[D'.] The partition function $Z^{D_k,G}({\bf t}),\, k\geq 4$ with $G=G_{D_k}$ is a particular tau-function of the DS hierarchy of $A_{2k-3}$-type satisfying \eqref{string}.
\item[D''.] The partition function $Z^{D_n^T,G}({\bf t}),\, n\geq 4$ with $G=G_{D_n^T}$ is a particular tau-function of the DS hierarchy of $D_{n}$-type satisfying \eqref{string}.
\item[E.] The partition function $Z^{E_n,G}({\bf t}),\,n=6,7,8,$ with $G=\langle J\rangle=G_{E_n}$ is a particular tau-function of the DS hierarchy of $E_{n}$-type satisfying \eqref{string}.
\end{itemize}
Summarizing, the partition function $Z^{X_k,G_{X_k}}({\bf t})$ with $X=A,D,D^T,$ or  $E$ 
is a particular tau-function of the DS hierarchy of $X_k^T$-type  satisfying \eqref{string}.
}

\medskip

In the case that $f=x^r$ with $G=\langle J \rangle=G_f$,  the FJRW invariants
$\langle\tau_{a_1 k_1}\cdots\tau_{a_N k_N} \rangle_g^{f,G}$ coincide with Witten's $r$-spin correlators.
The statement $A$ of Theorem-ADE justifies Witten's $r$-spin conjecture \cite{W2}, which was first proved by Faber--Shadrin--Zvonkine \cite{FSZ}; 
see ``Theorem $r$-spin''  below. 

For convenience of the reader let us recall some details  in the definition of Witten's $r$-spin correlators. 
For a given $N\geq 1$, let $a_1, \dots, a_N \in \{1,\dots,r\} $ be integers satisfying the following divisibility condition
\begin{equation}\label{divisible}
a_1+\cdots+a_N-N-(2g-2)=m r, \quad m\in\mathbb Z.
\end{equation}
For any smooth algebraic 
curve $C$ of genus $g$ with $N$ marked points $x_1$, \dots, $x_N$ there exists a line bundle ${\mathcal T}$ over $C$ such that
\begin{equation}\label{root}
{\mathcal T}^{\otimes r} =K_C\otimes{\mathcal O}\left((1-a_1)x_1\right)\otimes\cdots \otimes{\mathcal O}\left((1-a_N)x_N\right).
\end{equation}
Here $K_C$ is the canonical class of the curve $C$.
Moreover, there are $r^{2g}$ such line bundles. A choice of such an ``$r$-th root" of the bundle \eqref{root} defines a point in a covering of the moduli space. After a suitable compactification this covering is denoted by
\begin{equation}\label{root1}
p:\overline{\mathcal M}_{g,N}^{1/r}(a_1, \dots, a_N)\to \overline{\mathcal M}_{g,N}.
\end{equation} 
In genus zero, for a point $\left( C, x_1, \dots, x_N, {\mathcal T}\right)$ in the covering space, denote $V=H^1(C, {\mathcal T})$. 
This defines a vector bundle ${\mathcal V}\to \overline{\mathcal M}_{0,N}^{1/r}(a_1, \dots, a_N)$  because the space $V$ has constant dimension thanks to the fact that  $H^0(C, {\mathcal T})$ vanishes. 
Put
$$
c_W(a_1, \dots, a_N):= 
p_* \left( e\left({\mathcal V}^\vee\right) \right)\in H^{2(m-1)}\left( \overline{\mathcal M}_{0,N}\right),
$$
where $e\left({\mathcal V}^\vee\right)$ is the Euler class of the dual vector bundle ${\mathcal V}^\vee$. 
The
$c_W(a_1,\dots,a_N)$ is called the \textit{Witten class}. 
In higher genus, this is not completely correct because $H^0(C, {\mathcal T})$ is only {\it generically} zero and hence the vector bundle is only defined on a generic stratum. 
The Witten class $c_W(a_1,\dots,a_N)$ could still be defined as a particular 
cohomology class in $H^{2(m+g-1)}\left( \overline{\mathcal M}_{g,N}\right)$, but 
 the construction is more involved 
(see e.g. \cite{W2, FSZ, JKV,PV,Polishchuk}).
The $r$-spin intersection numbers are defined by
\begin{equation}\label{rspinn}
\left\langle \tau_{a_1 p_1}\cdots \tau_{a_N p_N} \right\rangle_g^{r-{\rm spin}} := 
\int_{\overline{\mathcal M}_{g,N}} c_W(a_1,\dots, a_N)  \psi_1^{p_1} \cdots \psi_N^{p_N}, \quad 
a_1,\dots,a_N \in \{ 1,\dots,r \}, ~ p_1,\dots,p_N \geq 0.
\end{equation}
 The  numbers $\left\langle \tau_{a_1 p_1} \cdots \tau_{a_N p_N} \right\rangle_g^{r-{\rm spin}}$ are zero unless
\be\label{dim-degree}
\frac{a_1-1}{r}+\cdots+\frac{a_N-1}r + \frac{r-2}r (g-1) + p_1+\cdots+p_N = 3g-3+ N.
\ee
The so-called \textit{Vanishing Axiom} conjectured in~\cite{JKV} and proven in~\cite{PV,Polishchuk} tells that the Witten class vanishes
if any of $a_i$, $i=1,\dots,N$ reaches~$r$. Hence, below, we only consider 
the case of $a_1,\dots,a_N$ belonging to $\{1,\dots, r-1\}$.

For computing Witten's $r$-spin correlators, we use Theorems~\ref{main thm}--\ref{topo-thm} for a particular tau-function along with the following result.

~~

\noindent \textbf{Theorem $r$-spin} (\cite{W2,FSZ}). 
\textit{
The partition function of $r$-spin intersection numbers
$$Z^{r-{\rm spin}}({\bf t}):=\exp\Biggl(\sum_{g,N\geq 0} \frac1{N!} \sum_{a_1,\dots,a_N=1}^n \sum_{k_1,\dots,k_N\geq 0}   
\langle\tau_{a_1 k_1} \cdots \tau_{a_N k_N}\rangle^{r-{\rm spin}}_g  t^{a_1}_{k_1} \cdots t^{a_N}_{k_N} \Biggr)$$
is a particular tau-function of the DS hierarchy of $A_n$-type, $n=r-1$ satisfying \eqref{string}. }

~~~

In \cite{LRZ}, Liu--Ruan--Zhang introduced {\it cohomological field theories with finite symmetry}, associated with simple singularities and certain symmetry groups, and with a $\Gamma$-invariant sector, where $\Gamma$ is the group of automorphisms of the Dynkin digram. These theories are proved to be related to the DS integrable hierarchies associated to the non-simply laced simple Lie algebras.

~~~

\noindent \textbf{Theorem-BCFG} (\cite{LRZ}).  \textit{The partition function of the $\Gamma$-invariant sector of $D_{n+1}^T,A_{2n-1},E_6$ FJRW theory with $G_{max}$ is a particular tau-function of the Drinfeld--Sokolov hierarchy of $B_n,C_n,F_4$-type
satisfying \eqref{string}; the partition function of the $\mathbb{Z}/3\mathbb{Z}$-invariant sector of $(D_4,\langle J\rangle)$ FJRW theory is a particular tau-function of the Drinfeld--Sokolov hierarchy of $G_2$-type satisfying \eqref{string}.}

~~~

Note that the common feature of Theorem-ADE and Theorem-BCFG claims that the partition function of FJRW invariants associated to a simple singularity with a symmetry group (possibly also with an invariant sector) is a tau-function of the DS hierarchy of $\g$-type, where $\g$ is a simple Lie algebra.  We call these numbers the {\it FJRW invariants of $\g$-type}, denoted by
$$\langle\tau_{a_1 k_1}\cdots\tau_{a_N k_N} \rangle_g^{FJRW-\g},\mbox{ or simply by } \langle\tau_{a_1 k_1}\cdots\tau_{a_N k_N} \rangle_g^\g.$$
As before, let $n$ denote the rank of $\g$. For a given $N\geq 1$ and for a collection of integers $a_1,\dots, a_N\in \{1,\dots,n\}$, we define the following generating functions of $N$-point  FJRW invariants of $\g$-type

\be\label{F-DS}
F_{a_1,\dots,a_N}^{FJRW}(\lambda_1,\dots,\lambda_N) : = (\kappa^{\frac1{h+1}} \sqrt{-h})^N 
\sum_{g,k_1,\dots,k_N\geq 0} \prod_{\ell=1}^N 
\frac{(-1)^{k_\ell} \left( \frac{m_{a_\ell}}{h}\right)_{k_\ell+1}}{\left( \kappa^{\frac1{h+1}} \, \lambda_\ell\right)^{\frac{m_{a_\ell}}{h}+k_\ell+1}} \langle\tau_{a_1 k_1} \cdots \tau_{a_N k_N}\rangle^{\g}_g.
\ee
Here $\kappa:=\le(\sqrt{-h}\ri)^{-h}.$ 

Combining the results of Theorems \ref{main thm} and \ref{topo-thm} with the statements of Theorem-ADE and Theorem-BCFG we arrive at the following formula for the FJRW invariants of $\mathfrak g$-type.

\bt \label{N-DS-p-real} 
Let $\g$ be a simple Lie algebra and $n$ the rank of $\g$. Let $M_a=M_a(\lambda)$, $a = 1, \dots, n$ be the generalized Airy resolvents of $\g$-type, which are the unique solutions to
\be \label{topo-g}
M'=  [M,\Lambda],
\ee
subjected to 
$$
M_a(\l)=\lambda^{-\frac{m_a}h} \Bigl[ \Lambda_{m_a}(\l)+ \mbox{lower degree terms w.r.t.~} \deg \Bigr].
$$
Here, $h$ is the Coxeter number and $m_a$ are the exponents of $\mathfrak g$. Then the generating functions \eqref{F-DS} for the $N$-point FJRW invariants of $\g$-type have the following expressions

\bea
&\& \!\!\!\!\! \frac{\d F_a^{FJRW}}{\d \lambda}(\lambda) =- \frac{1}{2\, h^\vee} B\Big(E_{-\theta}, M_a(\l)\Big) + \lambda^{-\frac{h-1}h}\, \delta_{a,n},\quad N=1, \label{N=1-FJR}\\
&\&  \!\!\!\!\!  F_{a_1,\dots,a_N}^{FJRW}(\lambda_1,\dots,\lambda_N) = - \frac 1 {2 N\, h^\vee} \sum_{s \in S_N} 
\frac {B\le(M_{a_{s_1}}(\lambda_{s_1}),\dots, M_{a_{s_N}}(\lambda_{s_N}) \ri) }{\prod _{j=1}^{N} (\lambda_{s_j}- \lambda_{s_{j+1}})} \nn\\
&\&  \qquad\qquad\qquad\qquad\qquad\quad -\delta_{N2} \, \eta_{a_1 a_2}  \frac {\lambda_1^{-\frac {m_{a_1}}h } \lambda_2^{- \frac {m_{a_2}}h} (m_{a_1} \,  \lambda_1 + m_{a_2} \, \lambda_2)}{(\lambda_1-\lambda_2)^2} , \quad N\geq 2. \label{24-FJR}
\eea
\et

 Eqs.\,\eqref{topo-g}--\eqref{24-FJR} are equivalent to the proposed 
formulae in~\cite{BDY2} (eq.~\eqref{previous-version-N} of the current paper). 
For other methods towards computing related invariants, see~\cite{BM} \cite{BY} \cite{BE} 
\cite{Buryak} \cite{CW} \cite{DZ-norm} \cite{DZ1} \cite{G1} \cite{LYZ} \cite{Zhou1}.

In particular, for given integers $r\geq 2$, $N\geq 1$ and a given collection of indices $a_1$, \dots, $a_N$ belonging to $\{1,\dots,r-1\},$ define 
\be\label{r-genera}F^{r-spin}_{a_1,\dots,a_N}(\lambda_1,\dots,\lambda_N)
: = \left(\kappa^{\frac1{r+1}}\,\sqrt{-r}\right)^N\sum_{k_1,\dots,k_N\geq 0}   \prod_{\ell=1}^N \frac{ (-1)^{k_\ell}   \left(\frac{a_\ell}{r}\right)_{k_\ell+1}}{ \left(\kappa^{\frac1{r+1}} \,  \lambda_\ell\right)^{\frac{{a_\ell}}r+k_\ell+1}} \langle\tau_{a_1 k_1} \cdots \tau_{a_N k_N}\rangle^{r{\rm-spin}}.
\ee
Here $\kappa=\le(\sqrt{-r}\ri)^{-r}$.  Note that we have omitted the genus labelling in the notation of correlator, since it can be obtained from the degree-dimension matching \eqref{dim-degree}.

\bt 
\label{r-spin-thm}
Let $n=r-1,\, \g= sl_{n+1} (\CC)$,  $\Lambda=\sum_{i=1}^n E_{i,i+1}+\lambda\, E_{n+1,1}$, and let $M_i=M_i(\lambda)$ be the basis of generalized Airy resolvents of $\g$-type, uniquely determined by the topological ODE 
\be\label{topo-An}
M'= [M,\Lambda],
\ee
subjected to
$$
M_a=\lambda^{-\frac a r} \left[ \Lambda^a+ \mbox{lower degree terms w.r.t. } \deg\right].
$$
Then the $N$-point functions \eqref{r-genera} of $r$-spin intersection numbers have the following expressions
\bea
&& \frac{\d F_a^{r-spin}}{\d \lambda}(\lambda)= -   (M_a)_{1,n+1}(\lambda) + \lambda^{-\frac{r-1}r}\, \delta_{a,n}, \quad N=1, \label{easy-1}\\
&& F_{a_1,\dots,a_N}^{r-spin}(\lambda_1,\dots,\lambda_N) = - \frac {1} N \sum_{s \in S_N} 
\frac {\Tr \le(M_{a_{s_1}}(\lambda_{s_1})  \dots  M_{a_{s_N}}(\lambda_{s_N}) \ri) }{\prod _{j=1}^{N} (\lambda_{s_j}- \lambda_{s_{j+1}})}  \nn\\
&& \qquad\qquad\qquad\qquad\qquad -\delta_{N2} \, \eta_{a_1 a_2}  \frac {\lambda_1^{-\frac {a_1}h } \lambda_2^{- \frac {a_2}h} ({a_1} \,  \lambda_1 + {a_2} \, \lambda_2)}{(\lambda_1-\lambda_2)^2},\quad\qquad N\geq 2. \label{r-spin-N-new}
\eea

\et
 
\begin{example}
[$r=2$] Witten's $2$-spin invariants coincide with intersection numbers of $\psi$-classes over $\overline{\mathcal{M}}_{g,N}$ \cite{W1,Kontsevich,FSZ}.
So Theorem~\ref{r-spin-thm} with the choice $r=2$
recovers the result of~\cite{BDY1,Zhou2}: 
\bea
&& \!\!\!\!\!\!\! {\sum_{g\geq 0}} \sum_{p_1,\dots,p_N\geq 0}  \frac{(2p_1+1)!!\cdots (2p_N+1)!!}{2^{2g-2+N}}\, \int\limits_{\overline{\mathcal{M}}_{g,N}} \psi_1^{p_1}\cdots \psi_N^{p_N}\,  \lambda_1^{-\frac{2p_1+3}{2}}\cdots\lambda_N^{-\frac{2p_N+3}2}\nn\\ 
&& \!\!\!\!\!\!\!  \qquad\qquad = - \frac{1}{N} \sum_{r\in S_{N}}  \frac{\Tr \left(M(\lambda_{r_1})\cdots 
M(\lambda_{r_N})\right)}{\prod_{j=1}^{N}(\lambda_{r_j}-\lambda_{r_{j+1}})} - \delta_{N2}\frac{\lambda_1^{-\frac12} \lambda_2^{-\frac12}(\lambda_1+\lambda_2)}{(\lambda_1-\lambda_2)^2}, \quad N\geq 2, \nn
\eea
where 
$$M=\frac{\lambda^{-\frac12}}{2}\left(
\begin{array}{cc}
- \frac 1 2 \sum_{g=1}^\infty \frac{(6g-5)!!}{96^{g-1}\cdot (g-1)!} \lambda^{-3g+2} & 2 \sum_{g=0}^\infty \frac{(6g-1)!!}{96^g\cdot g!} \lambda^{-3g}\\
\\
-2 \sum_{g=0}^\infty\frac{6g+1}{6g-1} \frac{(6g-1)!!}{96^g\cdot g!} \lambda^{-3g+1} &  
\frac 1 2 \sum_{g=1}^\infty \frac{(6g-5)!!}{96^{g-1}\cdot (g-1)!} 
\lambda^{-3g+2}\\
\end{array}
\right). $$
For $N=1$, it follows easily from \eqref{easy-1} the well-known formula 
$$\langle\tau_{3g-2}\rangle_g=\frac{1}{24^g g!}\quad \mbox{for}\quad g\geq 1.$$
\end{example}

\begin{example}[$r=3$]
We obtain from Theorem~\ref{r-spin-thm} that the only nontrivial one-point correlators have the following explicit expressions
\bea
\&\&
\int_{\overline{\mathcal M}_{3m-2,1}} 
c_W(1)\, \psi_1^{8m-7}=\frac1{6^{6m-4}(m-1)! \left(\frac13\right)_m}, \quad m\geq 1,\nn\\
\&\&
\int_{\overline{\mathcal M}_{3m,1}} c_W(2)\, \psi_1^{8m-2}
=\frac1{6^{6m}m! \left(\frac23\right)_m}, 
\quad m\geq 1.\nn
\eea
For $N\geq 2$, Witten's 3-spin correlators can be computed from the formulae
$$
F_{i_1,\dots,i_N}^{3-spin}(\lambda_1,\dots,\lambda_N) = - \frac {1} N \sum_{s \in S_N} 
\frac {\Tr\le(M_{i_{s_1}}(\lambda_{s_1})  \dots  M_{i_{s_N}}(\lambda_{s_N}) \ri) }{\prod _{j=1}^{N} (\lambda_{s_j}- \lambda_{s_{j+1}})} -\delta_{N2} \, \eta_{i_1 i_2}  \frac {\lambda_1^{-\frac {{i_1}}h } \lambda_2^{- \frac {{i_2}}h} ({i_1} \,  \lambda_1 + {i_2} \, \lambda_2)}{(\lambda_1-\lambda_2)^2}
$$
with explicit formulae of $M_a(\lambda)$ given in Appendix \ref{3spinM}.
\end{example}

\paragraph{Organization of the paper.} 
In Section~\ref{tau-section} we introduce the definition of tau-function
and prove Theorem~\ref{main thm}. 
In Section~\ref{essential} we define the essential series of~$\g$.
In Section~\ref{main proof}, we prove Theorem~\ref{topo-thm}.

\paragraph{Acknowledgements.} 
{We would like to thank the anonymous referee for constructive comments that helped  improve  the paper.}
We wish to thank Yassir Dinar, Daniele Valeri, Chao-Zhong Wu, Youjin Zhang for helpful discussions. 
D.\,Y. is grateful to Youjin Zhang for his advising. 
The work of M.\,B. is in part supported by the RGPIN/261229-2011 grant of the Natural
Sciences and Engineering Research Council of Canada  and by the FQRNT grant ``Matrices Al\'eatoires, Processus Stochastiques et Syst\`emes Int\'egrables" (2013--PR--166790). The work of D.\,Y. was 
initiated when he was a postdoctoral fellow  at SISSA, Trieste; he thanks SISSA for the excellent working conditions.

\setcounter{equation}{0}
\setcounter{theorem}{0}
\section{Tau-function of Drinfeld--Sokolov hierarchy} \label{tau-section}
\setcounter{equation}{0}
\setcounter{theorem}{0}
\subsection{Fundamental lemma}
Let $\g$ be a simple Lie algebra of rank $n,$ $L(\g)$ its loop algebra. Fix $\h$ a Cartan subalgebra of $\g$. 
We denote by $\rc \in \h$ the \textit{Weyl co-vector} of $\g$, which is uniquely determined by the following equations
\be \label{def-rho}
\alpha_i(\rc )=1,\qquad i=1,\dots,n.
\ee
Here $\alpha_i\in\h^*$ are simple roots. We define 
the \textit{principal} grading operator $\mbox{gr}$ on $L(\g)$ by
$$
\mbox{gr} = h \l\frac \d{\d \l} + \ad_{\rc}.
$$
It follows that $\deg a= j\in\mathbb{Z}$ {\it iff} $\mbox{gr}\, a =j\, a$, $\forall\, a\in \L(\g)$. 
We have the decomposition
$$ 
L(\g)=\bigoplus_{j\in\mathbb{Z}} L(\g)^j,\qquad\quad a\in L(\g)^j ~\Leftrightarrow ~ \mbox{gr}\, a =j\, a, \quad j\in \mathbb{Z}.
$$
For any $a\in L(\g)$, we denote its  principal decomposition by
$$a=\sum_{j\in \mathbb{Z}} a^{[j]},\qquad a^{[j]} \in L(\g)^j.$$

The following lemma is elementary but it will be frequently used.
\bl \label{van-Cartan}
Let  $x,y$ be any two elements in $\g=\g\otimes 1$ satisfying
${\rm gr}  \, x = k_1\, x, ~ {\rm gr} \,y = k_2\, y$.
If $k_1+k_2\neq 0$, then we have $(x|y)=0$.   
\el
\begin{proof} Suppose $k_1\neq 0$. By definition, ${\rm gr}\,x= k_1 \,x$ implies $[\rho^\vee, x]= k_1\,x$.  So we have
$$(x|y)= \frac1 {k_1} ( [\rho^\vee,x]\,|\,y)= - \frac1 {k_1} (x\,|\, [\rho^\vee,y])
=-\frac{k_2}{k_1} (x|y)~\Rightarrow~\frac{k_1+k_2}{k_1}(x|y)=0.$$
The lemma is proved.
\end{proof}

\begin{lemma}[fundamental lemma, \cite{DS}] \label{first-lemma} Let $q=q(x)$ be a $\fb$-valued smooth function, where 
$\fb:=\g^{\leq 0}.$ Let $\L=\p_x+\Lambda+ q(x)$.
Then there exists a unique pair $(U,H)$ of the form
\bea
U&=&\sum_{k\geq 1} U^{[-k]} (\lambda; q; q_x,\cdots) \in \mathcal{A}^q \otimes  {\rm Im} \, \ad_\Lambda,
 \label{U-form}\\
H&=&\sum_{j\in E_+} H^{[-j]} (\lambda; q; q_x,\cdots) \in \mathcal{A}^q \otimes {\rm Ker}  \, \ad_\Lambda, 
\label{H-form}
\eea
where ${\rm Im}, \, {\rm Ker}$ are taken in $\g((\lambda^{-1}))$, and $E_+:= \{ j\geq 0 \,|\,j\in E \}$, such that
\be \label{first-dressing}
e^{-\ad_U}\mathcal{L} = \p_x+\Lambda+H. 
\ee
\end{lemma}

\begin{proof}  
Eq.~\eqref{first-dressing} is equivalent to
$$
e^{-U}\circ \pa_x  \circ e^{ U} + e^{-\ad_U}\, (q+\Lambda) = \p_x+\Lambda + H.
$$
More explicitly this reads 
\be
\label{recurrUH}
\sum_{j=0}^{\infty} \frac {(-\ad_U)^j}{j!} \le( \frac {U_x}{j+1}  + q + \Lambda \ri) = \Lambda + H.
\ee
Comparing components with principal degree $-k$ of both sides of \eqref{recurrUH} we obtain
\be \label{recur}
H^{[-k]} + \left[U^{[-k-1]},\Lambda\right]=G_k\le(\lambda;q;U^{[-1]},\dots,U^{[-k]}; \p_x(U^{[-1]}),\dots,\p_x(U^{[-k]})\ri),\qquad k\geq 0.
\ee
Here, $G_k\in L(\g),\, k\geq 0.$ Moreover, entries of $G_k$ are polynomials in the entries of 
$$
 q,\,U^{[-1]},\dots,U^{[-k]},\, \p_x(U^{[-1]}),\dots,\p_x(U^{[-k]})
$$
whose coefficients are polynomials in $\lambda.$
The proof proceeds by induction on the principal degree. First, for $k=0$ eq.~\eqref{recur} reads
\be
H^{[0]}+\le[U^{[-1]} \,,\, \Lambda\ri]  = q^{[0]}.
\ee
Observe that an element $x\in\g$ has zero principal degree {\it iff} $x\in \h$. So $q^{[0]}$ belongs to~$\h$. 
Let us show that
$\h\subset \Im \, \ad_{\Lambda}$. This is equivalent to orthogonality 
\begin{equation}\label{orto}
\le(x \, | \, \Lambda_{m_a}\ri)=0  \quad \mbox{for any~} x\in \h, \quad a=1,\dots,n.
\end{equation}
Indeed, by Lemma~\ref{van-Cartan}, any element $y\in\g$ of nonzero principal degree is orthogonal to~$\h$. 
It remains to recall that any $\Lambda_{m_a}$ has the form $\Lambda_{m_a}=L_{m_a}+\lambda \, K_{m_a-h}$, 
where $L_{m_a}$ and $ K_{m_a-h}$ belong to~$\g$ and have nonzero principal degree. 
This proves orthogonality \eqref{orto}. So we have $H^{[0]}=0.$
Noting that the map $\ad_{\Lambda} :   {\rm Im} \, \ad_\Lambda \to {\rm Im} \, \ad_\Lambda$ is invertible, and  we have  
\be \label{start-U1}
 U^{[-1]}= \ad_\Lambda^{-1} (q^{[0]}) \in {\rm Im} \, \ad_\Lambda.
\ee 
The induction step clearly follows from eq.~\eqref{recur} and the decomposition 
$$
L(\g)=\Ker \, \ad_\Lambda \, \oplus \, {\rm Im} \, \ad_\Lambda.
$$
The lemma is proved.
\end{proof}

\begin{example}
Looking at equation~\eqref{recurrUH} with principal degree~$-1$, we have
$$
\HH^{[-1]} + \left[U^{[-2]},\Lambda\right]= \frac 1 2 \left[U^{[-1]},\left[U^{[-1]},\Lambda\right]\right] + \pa_x (U^{[-1]}) - \left[U^{[-1]},q^{[0]}\right] + q^{[-1]}.  
$$
Since $U^{[-2]}$ is assumed to be orthogonal to ${\rm Ker}\, \ad_\Lambda$, 
this equation uniquely determines $\HH^{[-1]}$ and $U^{[-2]}$ as indicated in the above proof.
\end{example}

\setcounter{equation}{0}
\setcounter{theorem}{0}
\subsection{$\g$-valued resolvents}

\bd \label{resol} 
Let $q=q(x)\in \fb$.
An element $R\in \mathcal{A}^q \otimes  \g((\lambda^{-1}))$ is called a {\bf  resolvent} of $\mathcal{L}$ if
\be
[\mathcal{L},R]=0.
\ee
The set of all resolvents of~$\L$ is denoted by $\mathcal{M}_\L$, called {\rm the resolvent manifold}.
\ed

For more about resolvents see for example \cite{BDY1} \cite{Dickey0} \cite{Dickey1} \cite{DS} \cite{GD}.

\bl [\cite{DS}] \label{decomp1} 
We have
$$
\mathcal{M}_\L=e^{\ad_U} \le( \Ker \, \ad_\Lambda\ri),
$$
where we note that the kernel\footnote{In the published version of this paper, 
the kernel is taken in $L(\g)$, so the resolvent manifold considered there 
is smaller and the homogeneity condition for Definition~\ref{resol-defi} 
(cf.~Proposition~\ref{boundary}) is not needed. The corrections made here and the addition of the 
homogeneity condition in Definition~\ref{resol-defi} are more consistent with Definition~\ref{resol}.} is taken in $\g((\lambda^{-1}))$. 
\el
\begin{proof}
Lemma \ref{first-lemma} reduces the problem to considering 
the resolvent manifold of $\p_x+\Lambda+H$.  So, let us look at the following equation for 
$R_H\in \mathcal{A}^q \otimes  \g((\lambda^{-1})):$
$$ [R_H, \p_x+\Lambda+H]=0.$$
Decompose 
$$
R_H= R_H^{\rm ker}+ R_H^{\rm im}, \qquad  R_H^{\rm ker} \in \mathcal{A}^q\otimes\Ker \, \ad_\Lambda,\,R_H^{\rm im} \in \mathcal{A}^q\otimes\Im \, \ad_\Lambda.
$$
It follows that 
$$ \frac{\p R_H^{\rm ker}}{\p x}+\frac{\p R_H^{\rm im}}{\p x} = \left[R_H^{\rm im}, \Lambda+H\right].$$
The right hand side of the above equation is in the image of $\ad_\Lambda$, so we have
\bea
 &\& \frac{\p R_H^{\rm ker}}{\p x}=0, \label{kereq}\\
 &\& \frac{\p R_H^{\rm im}}{\p x} = \left[R_H^{\rm im}, \Lambda+H\right].  \label{imeq}
\eea
Equation~\eqref{kereq} implies that 
$R_H^{\rm ker}$
can only depend on~$\lambda$. The rest is to show that 
$R_H^{\rm im}$
must vanish. If it does not vanish, then there exists an integer~$d$ such that 
$$
R_H^{\rm im} = \sum_{i=-\infty}^d R_H^{{\rm im},[i]},\qquad R_H^{{\rm im},[d]}\neq 0.
$$ 
Noting that $\deg\, H<0$, then looking at the highest degree term on both sides of eq.~\eqref{imeq} we obtain 
$$
\left[\Lambda, R_H^{{\rm im},[d]}\right]=0. 
$$
So we have $R_H^{{\rm im},[d]}=0$.
This produces a contradiction. 
The lemma is proved. 
\end{proof}

\begin{proposition}  \label{boundary}
There exist unique series $R_1,\dots,R_n$ satisfying the following system of equations 
\bea
&\& [\L,R_a]=0, \qquad R_a\in  \mathcal{A}^q \otimes  \g((\lambda^{-1})),  \label{ODE-R}\\
&\& R_a(\lambda;q,q_x,\dots)= \Lambda_{m_a} +\mbox{ lower order terms w.r.t. } \deg,
\label{boundary-R-1}\\
&\& \le(R_a(\lambda;q,q_x,\dots\ri)\, | \,R_b(\lambda;q,q_x,\dots)) = h \, \eta_{ab}\, \lambda,  \label{convention-R}
\eea
together with the requirements that $R_a$ are homogeneous of the extended principal degrees $m_a$.
\end{proposition}

This unique system of solutions $R_1$, \dots, $R_n$ is called in 
Section~\ref{intro-s} the basic resolvents of the operator~${\mathcal L}$.

\begin{proof}
The existence follows from the fact that $e^{\ad_U}(\Lambda_{m_a})$ is a solution,
where \eqref{convention-R} is due to \eqref{norm-Lambda-2}, and \eqref{boundary-R-1} is due to 
\eqref{U-form}. The uniqueness follows from Lemma \ref{decomp1}.
\end{proof}

\bc Let $U$ be defined as in Lemma \,\ref{first-lemma}.  Then
the basic resolvents~$R_a$ satisfy 
$$R_a= e^{\ad_U} (\Lambda_{m_a}),\qquad a=1,\dots, n.$$
\ec
From this  corollary we promptly deduce the 
following commutativity between the basic resolvents:
\be\label{commrr}
[R_a,R_b] =0.
\ee

\bd\label{basic-resol}   
Define
$
P_{m_a+hk} : =  \lambda^k  R_a= e^{\ad_U}  (\Lambda_{m_a+hk}),\quad k\geq 0.
$
\ed

The pre-DS hierarchy can be written as
$$
\frac{\p \L}{\p T^a_k}= \Bigl[\bigl(P_{m_a+kh}\bigr)_+ \,,\, \L \Bigr],\quad k\geq 0.
$$
As customary in the literature,
we will sometimes write $T^a_k$ as $T_{m_a+kh}$, $a=1,\dots,n,k\geq 0$. 

\bl   \label{pj-flow}  $\forall\, i,j\in E_+,$ we have
\bea
&\& \frac{\p P_j}{\p T_i}= \bigl[(P_i)_+\,,\,P_j\bigr], \label{zero-pre} \\
&\& \frac{\p  (P_i)_+}{\p T_j} -  \frac{\p  (P_j)_+}{\p T_i}+ \,
  \bigl[(P_i)_+ \,, \, (P_{j})_+\bigr]=0. \label{zero-curv}
\eea
\el
\begin{proof} 
Using the fundamental lemma~\ref{first-lemma}  we have
$$
\frac{\p \L}{\p T_i}= \bigl[ (P_i)_+ \, , \, \L \bigr]  
  \quad \Rightarrow \quad  \Bigl[ \p_{T_i} - (P_i)_+ \, , \, \L \Bigr]  =0 \quad 
\Rightarrow \quad  \Big[\p_{T_i} + S_i \,, \,  \p_x + \Lambda+H \Big]=0, \nn
$$
where $S_i:=\sum_{k=0}^\infty \frac{(-1)^k}{(k+1)!} \ad_U^k \left(\frac{\p U}{\p T_i}\right) - e^{-\ad_U} \le[(P_i)_+\ri].$ 
Clearly, $S_i$ takes values in {$\mathcal{A}^q\otimes \g((\lambda^{-1}))$.} 
Decompose 
$$
S_i= S_i^{\rm ker}+S_i^{\rm im},\qquad S_i^{\rm ker} \in \mathcal{A}^q \otimes \Ker\,\ad_\Lambda,\quad S_i^{\rm im} \in \mathcal{A}^q \otimes \Im\,\ad_\Lambda.
$$
Then we have
$$
\frac{\p H}{\p T_i } - \frac{\p S_i}{\p x}+ [S_i, \Lambda+H] =0 \quad \Rightarrow \quad  
\left\{\begin{array}{l} \frac{\p H}{\p T_i } - \frac{\p S_i^{\rm ker}}{\p x}=0, \\  \frac{\p S_i^{\rm im}}{\p x}= [S_i^{\rm im}, \Lambda+H].
\end{array}\right.
$$
Using the same argument as in the proof of Lemma \,\ref{decomp1} we find from the above equation for 
$S_i^{\rm im}$ that  $S_i^{\rm im}$ must vanish. So $S_i$ belongs to $\mathcal{A}^q\otimes \Ker\,\ad_\Lambda.$
On another hand, 
$$
 \frac{\p P_j}{\p T_i}=  \Bigl[(P_i)_+ ,P_j\Bigr] \quad \Leftrightarrow  \quad \Bigl[\p_{T_i}- (P_i)_+ , P_j\Bigr]=0
\quad \Leftrightarrow \quad  \Bigl[ \p_{T_i} + S_i, \Lambda_j \Bigr]=0.
$$
Hence eq.~\eqref{zero-pre} is proved.
Clearly eq.~\eqref{zero-pre} implies eq.~\eqref{zero-curv}; this is because
$$
\mbox{LHS of eq.~} \eqref{zero-curv} = \Bigl[(P_j)_+,P_i\Bigr]_+ 
-   \Bigl[(P_i)_+,  P_j\Bigr]_+ + \Bigl[(P_i)_+,(P_j)_+\Bigr]=0. 
$$
\end{proof}

\bl
\label{lem-gen-R} $\forall\,a=1,\dots,n$, we have
\be\label{qdef}
 \nabla_{a}(\lambda)  \, R_{b}(\mu) = \frac{[R_{a}(\lambda),R_{b}(\mu)] }{\lambda-\mu} - [Q_a, R_b(\mu)],\quad\quad Q_a:={\rm Coef} \bigl(R_a(\lambda),\lambda^1\bigr).  \label{gen-R}
\ee
\el
\begin{proof}  We have
\bea
\nabla_a(\lambda) \, R_b(\mu)&=&\sum_{k\geq 0} \frac{\p_{T^a_k} R_b(\mu)}{\l^{k+1}} = \sum_{k\geq 0} \frac{[ (\mu^k \, R_a(\mu))_+\,,\,R_b(\mu)]}{\l^{k+1}}\nn\\
&=& - \sum_{k\geq 0} \frac{[ \res{\rho=\infty} \frac{\rho^k R_a(\rho)}{\rho-\mu} d\rho \,,\,R_b(\mu)]}{\l^{k+1}} \nn\\
&=& \frac{1}{2\pi \sqrt{-1}} \oint_{|\mu|<|\rho|<|\lambda|}  d\rho \, \frac{[ R_a(\rho) \,,\,R_b(\mu)]}{(\lambda-\rho)(\rho-\mu)}\nn\\
&=& \frac{[R_a(\lambda),R_b(\mu)] }{\lambda-\mu} - \left[{\rm Coef} (R_a(\lambda),\lambda^1) \,,\,R_b(\mu)\right].\nn
\eea
\end{proof}

\setcounter{equation}{0}
\setcounter{theorem}{0}
\subsection{Two-point correlation functions}

Recall that in Definition \ref{tau-symm-def}, the two-point correlation functions $\Omega_{a,k;b,\ell}$ was defined by
\be\label{tau-symm-eq-again}
\sum_{k,\ell\geq 0} \frac{\Omega_{a,k;b,\ell}}{\lambda^{k+1} \mu^{\ell+1}} =  \frac{ \le(R_a (\lambda) \, | \, R_b(\mu)\ri)}{(\lambda-\mu)^2} -
\eta_{ab} \frac{m_a  \lambda + m_b \mu}{(\lambda-\mu)^2}.
\ee

\bl 
Definition \ref{tau-symm-def}, i.e., the above formula~\eqref{tau-symm-eq-again} is well-posed.
\el
\begin{proof}
Noting that\footnote{We would like to thank Anton Mellit for bringing our attention to the useful formula \eqref{Anton}.}
\be\label{Anton}
R_b(\mu) = R_b(\lambda)+ R_b' (\lambda)(\mu-\lambda) + 
(\mu-\lambda)^2 \, \p_\lambda \le(\frac{R_b(\lambda)-R_b(\mu)}{\lambda-\mu}\ri)
\ee
and using eqs.~\eqref{norm-Lambda-2} we have
$$
\frac{ \le(R_a (\lambda) \, | \, R_b(\mu)\ri)}{(\lambda-\mu)^2} = 
\eta_{ab} \, \frac{h\,\lambda}{(\lambda-\mu)^2} 
- \frac{ \le(R_a (\lambda) \, | \, R_b'(\lambda)\ri)}{\lambda-\mu} 
+ \le(R_a(\lambda)\, \Big|\, \p_\lambda \le(\frac{R_b(\lambda)-R_b(\mu)}{\lambda-\mu}\ri)\ri).
$$
In the above formulae, prime, $``\,'\,"$, denotes derivative w.r.t. the spectral parameter. 
Since $R_a(\lambda) = \mathcal{O}(\lambda^1)$, $a=1,\dots,n$, we know that the third term 
in the above identity has the form as the left hand side of~\eqref{tau-symm-eq}. Therefore it remains to show
$$\eta_{ab} \, \frac{h\,\lambda}{(\lambda-\mu)^2} 
- \frac{ \le(R_a (\lambda) \, | \, R_b'(\lambda)\ri)}{\lambda-\mu} - 
\eta_{ab} \frac{m_a  \lambda + m_b \mu}{(\lambda-\mu)^2}
$$
has the form as the left hand side of~\eqref{tau-symm-eq}. We will actually prove that the above expression vanishes.
Indeed,
\be\label{x-ind}
\p_x  \le(R_a (\lambda) \, | \, R_b'(\lambda)\ri) = \le([R_a (\lambda),\Lambda+q] \, | \, R_b'(\lambda)\ri) 
+\le(R_a (\lambda) | \le[R_b'(\lambda),\Lambda+q\ri]+ \le[R_b(\lambda),\Lambda'\ri]\ri) =0.
\ee
Here we have used the $\ad$-invariance of the Cartan--Killing form and the commutativity ~\eqref{commrr} between resolvents.
Noting that $R_a\in \mathcal{A}^q \otimes  \g((\lambda^{-1})),$
we find that \eqref{x-ind} implies that $\le(R_a (\lambda) \, | \, R_b'(\lambda)\ri)$ does not depend on 
$q,q_x,q_{2x},\dots$, i.e. it is 
just a function of $\lambda$. Hence 
$$\le(R_a (\lambda) \, | \, R_b'(\lambda)\ri)= \le(R_a (\lambda) \, | \, R_b'(\lambda)\ri)_{q(x)\equiv0} 
= \le(\Lambda_{m_a} \, | \, \Lambda_{m_b}'\ri).$$
The second equality uses \eqref{convention-R}. 
To compute $\le(\Lambda_{m_a} \, | \, \Lambda_{m_b}'\ri)$, as before, write 
$$\Lambda_{m_a} =  L_{m_a} +  \l  \, K_{m_a-h} , \quad L_{m_a} \in \g^{m_a} , ~ K_{m_a-h} \in \g^{m_a-h}, \quad a=1,\dots,n.$$
Using Lemma~\ref{van-Cartan} we have
$$\le(\Lambda_{m_a} \, | \, \Lambda_{m_b}'\ri)=\le(L_{m_a}\,| \,K_{m_b-h}\ri).$$
Note that $ (\Lambda_{m_a}\,|\,\Lambda_{m_b})=\eta_{ab}\,h\,\lambda$ implies that
\be \label{e1} 
(L_{m_a}\,|\,K_{m_b-h})+(L_{m_b}\,|\,K_{m_a-h})=\eta_{ab} \, h.
\ee
The commutativity 
$ [\Lambda_{m_a},\Lambda_{m_b}]=0$
implies that 
$$[K_{m_a-h}, L_{m_b}] + [L_{m_a}, K_{m_b-h}]=0.$$
Applying $(\rho^\vee \, | \, \cdot)$ to the above equation and using the $\ad$-invariance of $(\cdot|\cdot)$
we have
$$([\rho^\vee, K_{m_a-h}] \,| \, L_{m_b}) + ([\rho^\vee,L_{m_a}]\,|\, K_{m_b-h}])=0 ~
\Rightarrow ~ (m_a-h)\,(K_{m_a-h} \,| \, L_{m_b}) +  m_a \, (L_{m_a}\,|\, K_{m_b-h}])=0.$$
Combining eqs.~\eqref{e1} and the above equation we obtain
\be \label{LK}
 (L_{m_a}\,| \, K_{m_b-h}) = \eta_{ab} \, m_b,\quad \forall \, a,b.
\ee
Hence 
$$\eta_{ab} \, \frac{h\,\lambda}{(\lambda-\mu)^2} 
- \frac{ \le(R_a (\lambda) \, | \, R_b'(\lambda)\ri)}{\lambda-\mu} - 
\eta_{ab} \frac{m_a  \lambda + m_b \mu}{(\lambda-\mu)^2}=0.
$$
The lemma is proved.
\end{proof}

\bp \label{cor-1} The following formulae hold true
\be \label{gen-simple-2point-p}
\sum_{k\geq0} \frac{\Omega_{a,k;b,0}}{\lambda^{k+1}}=  
\le(R_a(\lambda) \,|\, Q_b  \ri) - \eta_{ab}\,m_b,\quad \forall\, a,b.
\ee
In particular, we have
\be \label{gen-ham-p}
\sum_{k\geq0} \frac{ \Omega_{a,k;1,0}}{\lambda^{k+1}}=  \le(R_a(\lambda) \,|\, E_{-\theta}\ri) - \eta_{a1}.
\ee
\ep
\begin{proof}
Taking in~\eqref{tau-symm-eq-again} the residue w.r.t.~$\mu$ at $\mu=\infty$ 
we obtain~\eqref{gen-simple-2point-p}. Noticing that
$$R_1(\mu)=\mu \, E_{-\theta}+I_+ + \mbox{ terms with principal degree lower than~} 1  $$
we must have $Q_1={\rm Coef}\le(R_1(\mu),\mu^1\ri)=E_{-\theta}$. This proves~\eqref{gen-ham-p}.
\end{proof}

\setcounter{equation}{0}
\setcounter{theorem}{0}
\subsection{Tau-function:  Proof of Lemmas \ref{tau-symmetry-def},~\ref{tau-symm-def-xt}}

We are ready to introduce our definition of tau-function. 
We begin with the proof of Lemma~\ref{tau-symmetry-def}.

\noindent \textit{Proof} of Lemma~\ref{tau-symmetry-def}. ~ First of all we have 
\bea
\sum_{k,\ell\geq 0} \frac{\Omega_{a,k;b,\ell}}{\lambda^{k+1} \mu^{\ell+1}} 
&= & \frac{ \le(R_a (\lambda) \, | \, R_b(\mu)\ri)}{(\lambda-\mu)^2} -
\eta_{ab} \frac{m_a  \lambda + m_b \mu}{(\lambda-\mu)^2} =  \frac{ \le(R_b(\mu)\,  | \, R_a (\lambda) \ri)}{(\mu-\lambda)^2} -
\eta_{ba} \frac{m_b \mu + m_a  \lambda}{(\mu-\lambda)^2}\nn\\
&= & \sum_{k,\ell\geq 0} \frac{\Omega_{b,k;a,\ell}}{\mu^{k+1} \l^{\ell+1}} =
 \sum_{k,\ell\geq 0} \frac{\Omega_{b,\ell;a,k}}{\mu^{\ell+1} \l^{k+1}}, \nn
\eea
where we have used the symmetry property of~$\eta_{ab}$ and~$(\cdot|\cdot)$.
It follows $\Omega_{a,k; b,\ell}=\Omega_{b,\ell; a,k}$.

Secondly, by using Lemma~\ref{lem-gen-R} we have
\bea
\sum_{k,\ell,m\geq 0}  \frac{\p_{T^c_m} \, \Omega_{a,k;b,\ell} }{ \xi^{m+1} \lambda^{k+1} \mu^{\ell+1}} 
&=& \nabla_c(\xi)\,  \sum_{k,\ell\geq 0}  
\frac{\Omega_{a,k;b,\ell}}{\lambda^{k+1} \mu^{\ell+1}} 
\nn\\
&= & \frac{  \le(\nabla_c(\xi) \, R_a (\lambda) \, | \, R_b(\mu)\ri)}{(\lambda-\mu)^2}  
+ \frac{  \le(R_a (\lambda) \, | \, \nabla_c(\xi)\,  R_b(\mu)\ri)}{(\lambda-\mu)^2} \nn\\
&= & \frac{  \le([R_c(\xi),R_a(\lambda)] \, | \, R_b(\mu)\ri)}{(\lambda-\mu)^2(\xi-\lambda)}  - 
 \frac{  \le([Q_c ,R_a (\lambda)] \, | \, R_b(\mu)\ri)}{(\lambda-\mu)^2} \nn\\
&\& + \frac{  \le(R_a (\lambda) \, | \, [R_c(\xi),R_b(\mu)]\ri)}{(\lambda-\mu)^2(\xi-\mu)}
-  \frac{  \le(R_a (\lambda) \, | \, [Q_c,P_b(\mu)]\ri)}{(\lambda-\mu)^2}. \nn
\eea
Clearly the two terms with negative signs give a zero contribution due to 
the $\ad$-invariance of the Cartan--Killing form. The remaining two terms simplify to
$$
 \frac{\le([R_c(\xi),R_a(\lambda)] \, | \, R_b(\mu)\ri)}{(\lambda-\mu)^2}  \le( \frac{1}{\xi-\lambda}  -
\frac{1}{\xi-\mu}\ri) = - \frac{\le([R_c(\xi),R_a(\lambda)] \, | \, R_b(\mu)\ri)} {(\lambda-\mu)(\mu-\xi)(\xi-\lambda)}.
$$
So we have
$$
\sum_{k,\ell,m\geq 0}  \frac{\p_{T^c_m} (\Omega_{a,k;b,\ell})  }{\xi^{m+1} \lambda^{k+1} \mu^{\ell+1}}=
- \frac{\le([R_c(\xi),R_a(\lambda)] \, | \, R_b(\mu)\ri)} {(\lambda-\mu)(\mu-\xi)(\xi-\lambda)}.
$$
This gives also
$$\sum_{k,\ell,m\geq 0}  \frac{\p_{T^a_k} (\Omega_{c,m;b,\ell})  }{ \lambda^{k+1}  \xi^{m+1}  \mu^{\ell+1} } = 
- \frac{\le([R_a(\lambda),R_c(\xi)] \, | \, R_b(\mu)\ri)} {(\xi-\mu)(\mu-\lambda)(\lambda-\xi)}.$$
Hence 
\be \label{skewsymmetryddd}
\p_{T^c_m} (\Omega_{a,k;b,\ell}) = \p_{T^a_k} (\Omega_{c,m;b,\ell})
\ee
due to skew-symmetry of the Lie bracket. The lemma is proved. $\hfill \Box$

\medskip

\noindent\textit{Proof} of Lemma~\ref{tau-symm-def-xt}.  ~ 
It suffices to show the compatibility between~\eqref{xt10} and~\eqref{two-point-tau}, namely, to show that
\be\label{toshowak10}
\frac{\p\Omega_{a,k;b,\ell}}{\p T^{1,0}} = - \frac{\p\Omega_{a,k;b,\ell}}{\p x}.
\ee
Taking $c=1,m=0$ in the already proved identity~\eqref{skewsymmetryddd} we have
$$
\p_{T^a_k} (\Omega_{1,0;b,\ell}) = \p_{T^1_0} (\Omega_{a,k;b,\ell}).
$$
Hence \eqref{toshowak10} is equivalent to  
$$
\frac{\p \Omega_{1,0;b,\ell}}{\p T^{a,k}} = - \frac{\p\Omega_{a,k;b,\ell}}{\p x}.
$$
Let us make a generating function. Then the above identity is equivalent to
$$
\sum_{k,\ell} \frac{\p \Omega_{1,0; b,\ell}}{\p T^{a,k}} z^{-k-1} w^{-\ell-1} = - \sum_{k,\ell} \frac{\p\Omega_{a,k;b,\ell}}{\p x} z^{-k-1} w^{-\ell-1}.
$$
We have
\bea
- {\rm RHS}
&=& \frac{(\p_x R_a(z) | R_b(w))}{(z-w)^2}+ \frac{  ( R_a(z) | \p_xR_b(w))}{(z-w)^2} \nn\\
&=& \frac{([R_a(z),\Lambda(z)+q] | R_b(w))}{(z-w)^2}+ \frac{( R_a(z) | [ R_b(w), \Lambda(w)+q])}{(z-w)^2} \nn\\
&=& \frac{(\Lambda(z)+q | [R_b(w), R_a(z)])}{(z-w)^2} - \frac{( \Lambda(w)+q | [ R_b(w), R_a(z)])}{(z-w)^2}\nn \\
&=& \frac{(\Lambda(z) - \Lambda(w) | [R_b(w), R_a(z)])}{(z-w)^2}. \nn
\eea
Recall that $$\Lambda(z)= I_++ z E_{-\theta},\qquad \Lambda(w)= I_++ w E_{-\theta}.$$
So we have 
$$
- {\rm RHS}=\frac{((z-w)E_{-\theta}|[R_b(w), R_a(z)])}{(z-w)^2} = \frac{(E_{-\theta}|[R_b(w), R_a(z)])}{z-w}.
$$
On another hand, we have
\bea
{\rm LHS} 
&=& \nabla_a(z) \, \sum_l  \Omega_{1,0; b,l} \, w^{-l-1} \nn\\ 
&=& \nabla_a(z) \, \left[ (E_{-\theta} | R_b(w))+ \mbox{const}\right] \nn\\
&=&  \le(E_{-\theta} | \nabla_a(z) \left[R_b(w)\right] \ri) \nn\\
&=& \frac{ \le(E_{-\theta} | [R_a(z),R_b(w)] \ri)}{z-w}  +  \le(E_{-\theta} | [Q_a,R_b(w)] \ri). \nn
\eea
We note that the second term of the last expression must be zero because
\be\label{simple}
\deg Q_a+ h \leq m_a \quad \Rightarrow \quad [E_{-\theta}, Q_a]=0.
\ee
The lemma is proved.  $\hfill\Box$

Hence we have arrived at our definition of tau-function, i.e., Definition~\ref{our-def-tau}.

\setcounter{equation}{0}
\setcounter{theorem}{0}
\subsection{Gauge invariance}
In this subsection, we show that the tau-function in Definition \ref{our-def-tau} is, in fact, gauge invariant. Recall that the change of the Lax operator
\be
\mathcal{L}=\partial_x +\Lambda+q(x) \quad  \mapsto \quad 
\widetilde{\L} = e^{\ad_{N(x)}} \mathcal{L}=\partial_x+\Lambda+\tilde q(x),\qquad N(x)\in \fn
\ee 
is called a {\it  gauge transformation}: $q\mapsto \tilde q$.  
It will also be convenient to deal with the 
infinitesimal form of~\eqref{gauge1}, $\widetilde {\mathcal L}={\mathcal L}+\delta{\mathcal L}$,
\be\label{gauge2}
\delta{\mathcal L} := \left[ N(x), {\mathcal L}\right]=\Bigl[ N(x), q(x)+I_+\Bigr]- \frac{\p N(x)} {\p x}.
\ee

Let $\widetilde R_a$, $a=1,\dots,n$ be the basic resolvents of~$\widetilde \L$.
It is not difficult to verify that $\widetilde{R}_a = e^{\ad_{N(x)}} R_a$. 

\bl \label{gauge-inv-lemma}
The gauge transformations \eqref{gauge1} are symmetries of the pre-DS hierarchy.
\el
\begin{proof} 
We have to prove the commutativity
$$
\frac{\p}{\p s} \frac{\p {\mathcal L}}{\p T} = \frac{\p}{\p T} \frac{\p {\mathcal L}}{\p s} 
$$
between the $j$-th flow of the pre-DS hierarchy
$$
\frac{\p\L}{\p T_j} =\left[ \left(P_j\right)_+\,,\, {\mathcal L}\right] ,\qquad  j\in E_+
$$
and the flow given by the infinitesimal gauge transformation
$$
\frac{\p \mathcal L}{\p s} =\left[ N, \L \right] 
$$
for some $\mathfrak n$-valued function $N=N(x)$. Using \eqref{invres} we derive
$$
\frac{\p P_j}{\p s} =\left[ N,P_j\right].
$$
So, after simple calculations with the help of the Jacobi identity we compute the difference between the mixed derivatives
$$
\frac{\p}{\p s} \frac{\p {\mathcal L}}{\p T} - \frac{\p}{\p T} \frac{\p {\mathcal L}}{\p s} =\left[ \left[ N, P_j \right]_+ - \left[ N, \left( P_j \right)_+\right], {\mathcal L}\right] =0.
$$
\end{proof}

The two-point correlation functions $\wt \Omega_{a,k;b,\ell},\,k,\ell\geq 0$ associated to $\wt \L$ 
are defined by
\be
\sum_{k,\ell\geq 0} \frac{ \wt \Omega_{a,k;b,\ell}}{\lambda^{k+1} \mu^{\ell+1}} =  
\frac{ \le(\wt R_a (\lambda) \, | \, \wt R_b(\mu)\ri)}{(\lambda-\mu)^2} -
\eta_{ab} \frac{m_a  \lambda + m_b \mu}{(\lambda-\mu)^2}.
\ee

\bl  \label{gauge-inv-g}
$\forall\, a,b$, $\forall\,k,\ell\geq 0,$  we have  $\wt \Omega_{a,k;b,\ell}=\Omega_{a,k;b,\ell}$.
\el
\begin{proof}
$\le(\wt R_a (\lambda) \, \big| \, \wt R_b(\mu)\ri) = \le(e^{\ad_{N(x)}} R_a(\lambda)\, \big| \, e^{\ad_{N(x)}} R_a(\mu)\ri) = \bigl( R_a (\lambda) \, \big| \,  R_b(\mu)\bigr). $
\end{proof}

\noindent  In a similar way one can easily prove that $\forall\,N\geq 2$ the correlation functions $\langle\langle\tau_{a_1 k_1} \dots \tau_{a_N k_N}\rangle\rangle^{DS}$ are gauge invariant. 

Now we are ready to prove Lemma \ref{tau-gauge-inv-lemma}.

\noindent \textit{Proof} of Lemma \ref{tau-gauge-inv-lemma}. \quad The lemma can be proved by using Lemma \ref{gauge-inv-g} and Definition \ref{our-def-tau}. $\hfill\Box$

Due to Lemma \ref{tau-gauge-inv-lemma}, $\forall\,N\geq 3$ the correlation functions $\langle\langle\tau_{a_1 k_1} \dots \tau_{a_N k_N}\rangle\rangle^{DS}$ are gauge invariant.

\setcounter{equation}{0}
\setcounter{theorem}{0}
\subsection{Gauge fixing and Drinfeld--Sokolov hierarchy}

We consider in this section a particular family of gauges {\cite{DS,BFRFW,DLZ}}.
\bd
A linear subspace $\V\subset\fb$ is called a \textit{gauge of DS-type} if 
$\fb=\V \oplus [I_+, \fn].$
\ed
\noindent 
Let $\V$ be a gauge of DS-type.
The fact that $\ad_{I_+}: \fn \rightarrow \fb$ is injective implies $\dim_{\CC}\, \V = n$. 
Write
$$
\V=\bigoplus_{j=-(h-1)}^0 \V^{j},\qquad \V^j\subset \g^j.
$$
Denote $\fb^j=\fb\cap \g^j.$ We have $\fb^j=\V^j\oplus [I_+,\fb^{j-1}],$ $j=-(h-1),\dots,0$. Clearly,
$\V^{-(h-1)}=\mathbb{C} E_{-\theta}$.
Noticing that for $j=-(h-1),\dots,0,$ the dimension $\dim \fb^{j}$ can be different from 
$\dim \fb^{j-1}$ {\it iff} $-j$ is an exponent of $\g$ \cite{Vara,DS}, we find that $\V^j$ is a null space unless $(-j)$ is an exponent. Thus
$$\V=\bigoplus_{a=1}^n V_{a},\quad \dim_{\CC} \,V_{a}=1,$$ where 
non-zero elements in $V_a$ have principal degree $-m_a$. 
We now take a basis $\{X^1,\dots, X^n\}$ of~$\V$ satisfying $\deg X^a= - m_a$.  
It has been proved in~\cite{DS} that
for any Lax operator $\L=\p_x+\Lambda+q(x)$, 
there exists a \textit{unique} $N^{\rm can}(x)\in \mathcal{A}_q\otimes \fn$
such that
\be\label{lcan}
e^{\ad_{N^{\rm can}(x)}}  \L = \p_x+ \Lambda + q^{\rm can}(x)=:\L^{\rm can},\quad \mbox{ for some } \V\mbox{-valued function } q^{\rm can}.
\ee
Write $q^{\rm can}= \sum_{a=1}^n  w_a \, X^a=(w_1,\dots,w_n)$.
The DS-flows of $q^{\rm can},$ or say of $w_a$, can be written as 
\be\label{DS-qcan-w-precise}
\frac{\p q^{\rm can}}{\p T^a_k} \,=\, \left[ \le(\lambda^k R_a^{\rm can}\ri)_+ \,, \, \L^{\rm can} \right] 
\,+\, \le[\frac{\p e^{N^{\rm can}}}{\p T^a_k}  e^{-N^{\rm can}} \,, \, \L^{\rm can}\right].
\ee
A priori, the right hand side of~\eqref{DS-qcan-w-precise} has a dependence on~$q$, 
as we can see from the second term that it contains 
flow of components of~$N^{\rm can}$. However, Lemma~\ref{gauge-inv-lemma} says that
the gauge transformation is a symmetry of the pre-DS hierarchy. So the right hand side of~\eqref{DS-qcan-w-precise} 
depends only on~$q^{\rm can}$, i.e., $w_a$, $a=1,\dots,n$ satisfy equations of the form
\be\label{DS-qcan-w}
\frac{\p w_a}{\p T^b_k} = G_{a,b,k}\left(q^{\rm can},q^{\rm can}_x,q^{\rm can}_{xx},\dots\right), \quad k\geq 0.
\ee 
\bd
Equations~\eqref{DS-qcan-w} are called the \textbf{DS hierarchy of $\g$-type} associated to~$\V$. 
\ed

Let~$R^{\rm can}_a$ be the basic resolvents of~$\L^{\rm can}$, 
and $\Omega^{\rm can}_{a,k;b,\ell}$ 
the two-point correlations functions of~$\L^{\rm can}$, i.e.,
\be
\sum_{k,\ell\geq 0} \frac{\Omega_{a,k;b,\ell}^{\rm can}}{\lambda^{k+1} \mu^{\ell+1}} =  
\frac{ \le(R_a^{\rm can} (\lambda) \, | \, R_b^{\rm can}(\mu)\ri)}{(\lambda-\mu)^2} -
\eta_{ab} \frac{m_a  \lambda + m_b \mu}{(\lambda-\mu)^2}.
\ee
\bc
Let~$\tau(\bdT)$ be a tau-function of the DS hierarchy. The following formulae hold true
$$
 \frac{\p^2 \log \tau}{\p T^a_k \p T^b_\ell} = \Omega_{a,k;b,\ell}^{\rm can},\quad \forall \, a,b=1,\dots,n, ~ k,\ell\geq 0.
$$
\ec
\begin{proof}
By gauge invariance of two-point correlation functions. 
\end{proof}
We also call~$\tau(\bdT)$ a tau-function of the solution 
$q^{\rm can}(\bdT)= \bigl(w^1(\bdT),\dots,w^n(\bdT)\bigr)$.

\setcounter{equation}{0}
\setcounter{theorem}{0}
\subsection{Proof of Theorem \ref{main thm}} \label{DS-N}
The proof will be almost identical to the proof for the case $\g=A_1$ case~\cite{BDY1}. 

\noindent {\it Proof} of Theorem \ref{main thm}. ~ 
For any permutation $s=[s_1,\dots,s_p]\in S_p$, $p\geq 2$, define
$$
P(s):= -\prod_{j=1}^p \frac{1}{\lambda_{s_j}-\lambda_{s_{j+1}}},\qquad 
\l_{s_{p+1}} \equiv \l_{s_1}.
$$
We first prove 
the generating formula of multi-point correlation functions of a solution of the
pre-DS hierarchy, then we use the ${\rm ad}$-invariance of $B$ for the gauge-fixed case.

Let $\L=\p_x+\Lambda+q(x),\,q(x)\in \fb$ be a linear operator,
$R_a$ the basic resolvents of $\L$. For an arbitrary solution 
$q(x,\bdT)$ to the pre-DS hierarchy \eqref{pre-DS}, let $\tau(\bdT)$ be the corresponding tau-function,
and $F_{a_1,\dots,a_N}(\bdT),\,N\geq 1$ 
the generating series of $N$-point 
correlations functions of $\tau(\bdT)$. 

We now use mathematical induction to prove formula \eqref{24} with $R^{\rm can}$ replaced by $R$.
For $N=2$, the formula is true by definition. 
Suppose it is true for $N=p\,(p\geq 2)$, then for $N=p+1,$ we have
\bea
&\& F_{\alpha_1,\dots,\alpha_{p+1}}(\lambda_1,\dots,\lambda_{p+1}; \TT)= 
\nabla_{\alpha_{p+1}}(\lambda_{p+1})\,F_{\alpha_1,\dots,\alpha_{p}}(\lambda_1,\dots,\lambda_{p};\TT)\nn\\
&=&
 - \frac {1} {2\, h^\vee\, p } \, \nabla_{\alpha_{p+1}}(\lambda_{p+1})\,  \sum_{\s \in S_p} 
\frac { B \le(R_{\alpha_{s_1}}(\lambda_{s_1}),\dots, R_{\alpha_{s_p}}(\lambda_{s_p})\ri)}{\prod_{j=1}^{p} (\lambda_{s_j} - \lambda_{s_{j+1}})}\nn\\
&=&- \frac {1} {2\,h^\vee\,p} \,  \sum_{s \in S_p} \sum_{q=1}^p
\frac { B
\le(R_{\alpha_{s_1}}(\lambda_{s_1}), \dots,  \le[\frac{R_{\alpha_{p+1}}(\lambda_{p+1})}{\l_{p+1} - \l_{s_q}}  + Q_{\alpha_{p+1}},R_{\alpha_{s_q}}(\lambda_{s_q})\ri],\dots, R_{\alpha_{s_p}}(\lambda_{s_p})\ri)}{\prod_{j=1}^{p} (\lambda_{s_j} - \lambda_{s_{j+1}})}.\nn
\eea
Recall that the elements $Q_\alpha\in\mathfrak g$ were defined in eq.~\eqref{qdef}.
Now we observe that the terms containing the commutator with $Q_{\alpha_{p+1}}$ sum up to zero due to the  $\ad$--invariance of $B$, namely due to the formula 
$$
\sum_{q=1}^p \Bigl(X_1,\dots, [A,X_q], X_{q+1}, \dots, X_p\Bigr) = 0\ , \ \ \forall \, X_1,\dots, X_p, A\in \g.
$$
Thus we are left with
\bea
&=& \frac {1} {2\,h^\vee\,p } \,  \sum_{s \in S_p} P(s) \sum_{q=1}^p \le( 
\frac {  
B\le(R_{\alpha_{s_1}}(\lambda_{s_1}) ,\dots, R_{\alpha_{s_{q-1}}}(\lambda_{s_{q-1}}) , R_{\alpha_{p+1}}(\lambda_{p+1}) , R_{\alpha_{s_q}}(\lambda_{s_q}) ,\dots, R_{\alpha_{s_p}}(\lambda_{s_p})\ri)}{\lambda_{p+1}-\lambda_{s_q}}\ri.\nn\\
&& \quad \le.-
\frac {  
B\le(R_{\alpha_{s_1}}(\lambda_{s_1}) ,\dots, R_{\alpha_{s_{q-1}}}(\lambda_{s_{q-1}}) , R_{\alpha_{p+1}}(\lambda_{p+1}) , R_{\alpha_{s_q}}(\lambda_{s_q}), \dots, R_{\alpha_{s_p}}(\lambda_{s_p})\ri)}{\lambda_{p+1}-\lambda_{s_{q-1}}} \ri)\nn\\
&=& \frac {1} {2\,h^\vee\,p} \,  \sum_{s \in S_p} P(s) \sum_{q=1}^p (\lambda_{s_q}-\lambda_{s_{q-1}})\nn\\
&&\qquad\qquad\qquad
\frac { 
B \le(R_{\alpha_{p+1}}(\lambda_{p+1}) ,R_{\alpha_{s_q}}(\lambda_{s_q}) ,\dots, R_{\alpha_{s_p}}(\lambda_{s_p}),R_{\alpha_{s_1}}(\lambda_{s_1}) ,\dots, R_{\alpha_{s_{q-1}}}(\lambda_{s_{q-1}})\ri)}{(\lambda_{p+1}-\lambda_{s_q})(\lambda_{p+1}-\lambda_{s_{q-1}})}\nn\\
&=&\frac1 {2\,h^\vee\,p} \sum_{q=1}^p \,  \sum_{s \in S_p} P([p+1,s_q,\dots,s_p,s_1,\dots,s_{p-1}])\nn\\
&&\qquad\qquad\qquad
B\le(R_{\alpha_{p+1}}(\lambda_{p+1}) , R_{\alpha_{s_q}}(\lambda_{s_q}) ,\dots, R_{\alpha_{s_p}}(\lambda_{s_p}), R_{\alpha_{s_1}}(\lambda_{s_1}) ,\dots, R_{\alpha_{s_{q-1}}}(\lambda_{s_{q-1}})\ri)\nn\\
&=& \frac{1}{2\,h^\vee}\, \sum_{s \in S_p} P([p+1,s])
\, B\le(R_{\alpha_{p+1}}(\lambda_{p+1}) , R_{\alpha_{s_1}}(\lambda_{s_1}) , \dots, R_{\alpha_{s_p}}(\lambda_{s_p})\ri).\nn
\eea

For any gauge~$\V$ of DS-type, there exists a unique $N(x)\in \mathcal{A}_q\otimes \fn$
such that 
$$e^{\ad_{N(x)}}  \L = \L^{\rm can}.$$
Observing that $\widetilde{R}_a = e^{\ad_{N(x)}} R_a$ and using the ${\rm Ad}$-invariance of $B$ we obtain 
$$
F_{a_1,\dots,a_N}(\lambda_1,\dots,\lambda_N;\bdT) = -\sum_{s \in S_N} 
\frac {B\le(R^{\rm can}_{a_{s_1}}(\lambda_{s_1}),\dots, R^{\rm can}_{a_{s_N}}(\lambda_{s_N})\ri)}
{2\,N\,h^\vee\, \prod _{j=1}^{N} (\lambda_{s_j}- \lambda_{s_{j+1}})} -\delta_{N2} \, \eta_{a_1 a_2}  \frac {m_{a_1} \,  \lambda_1 + m_{a_2} \, \lambda_2}{(\lambda_1-\lambda_2)^2  }.
$$
Finally, $F_{a_1,\dots,a_N}(\lambda_1,\dots,\lambda_N;\bdT)\in \mathcal{A}^{q^{\rm can}} [[\lambda_1^{-1}, \dots, \lambda_N^{-1}]]$
due to Lemma~\ref{boundary} (with $\L$ replaced by $\L^{\rm can}$).   The theorem is proved.
\hfill$\Box$

\bc \label{cor-ham} Let $\V$ be a gauge of DS-type. For an arbitrary solution $q^{\rm can}$ to the DS hierarchy of $\g$-type associated to 
$\V$, let $\tau$ be the tau-function of this solution.
The following formulae hold true
\be \label{gen-simple-2point}
\sum_{k\geq0} \frac{\langle\langle \tau_{a,k} 
\tau_{b,0} \rangle\rangle^{DS}}{\lambda^{k+1}}=  \le(R_a^{\rm can}(\lambda) \,|\, Q_b^{\rm can}  \ri) -  \eta_{ab}\,m_b,\quad a,b=1,\dots,n.
\ee
In particular, we have
\be \label{gen-ham}
\sum_{k\geq0} \frac{\langle\langle \tau_{a,k} 
\tau_{1,0} \rangle\rangle^{DS}}{\lambda^{k+1}}=  \le(R_a^{\rm can}(\lambda) \,|\, E_{-\theta}\ri) - \eta_{a1},\quad a=1,\dots,n.
\ee
\ec
\begin{proof}
Taking in~\eqref{24} with $N=2$ the residue w.r.t.~$\mu$ at $\mu=\infty$ 
we obtain \eqref{gen-simple-2point}. To show \eqref{gen-ham}, we only need to notice that 
for $b=1$, $\hbox{Coeff}(R_1^{\rm can}(\mu),\mu^1)=E_{-\theta}$. Indeed,
$$R_1^{\rm can}(\mu)=\lambda \, E_{-\theta}+I_+ + \cdots. $$
Here, the dots denote terms with principal degree lower than~$1$ which contain no more $\lambda^1$-power. 
\end{proof}

More explicitly, let $(U^{\rm can},H^{\rm can})$ be the unique pair associated to~$\L^{\rm can}$. Note that
 \be
 R^{\rm can}_a = e^{\ad_{U^{\rm can}}} \Lambda_{m_a}.
 \ee
Also note that eq.~\eqref{U-form} implies that~$U^{\rm can}$ must have the following decomposition 
$$U^{\rm can}=\sum_{k\geq 0} U_{-k}^{\rm can}\, \lambda^{-k}, 
\qquad U_0^{\rm can}\in \fn,~U_{-k}^{\rm can}\in\g, ~k\geq 1.$$
Hence we have
\be
Q_b^{\rm can}=\hbox{ Coeff} (R_b^{\rm can}(\mu),\mu^1) = e^{\ad_{U_0^{\rm can}}} \, K_{m_b-h},\quad b=1,\dots,n.
\ee

Before ending this section, we consider taking a faithful irreducible matrix realization~$\pi$ of~$\g$. 
Let~$\chi$ be the unique constant satisfying
\be
(a|b)=\chi \, \Tr (\pi(a) \pi(b)),\qquad \forall \, a,b\in \g.
\ee
For simplicity we will write $\pi(a)$ just as~$a$, for $a\in \g.$ Similarly as Theorem~\ref{main thm} we have
\bp\label{273}
Let $\V$ be a gauge of DS-type, 
$\L^{\rm can}$ the gauge fixed Lax operator~\eqref{lcan}, and $R_a^{\rm can}$, $a=1,\dots, n$ the basic 
resolvents of~$\L^{\rm can}$.
For an arbitrary solution $q^{\rm can}(\bdT)$ to the DS hierarchy associated to~$\V$, we have
\bea
&& \!\!\!\!\!\!\!\! F_{a_1,\dots,a_N}(\lambda_1,\dots,\lambda_N;\bdT) = - \frac {1} {\chi  N} \sum_{s \in S_N} 
\frac {\Tr \,\, R_{a_{s_1}}^{\rm can}(\lambda_{s_1}) \cdots R^{\rm can}_{a_{s_N}}(\lambda_{s_N})  }{\prod _{j=1}^{N} (\lambda_{s_j}- \lambda_{s_{j+1}})}-\delta_{N2} \, \eta_{a_1 a_2}  \frac {m_{a_1}  \lambda_1 + m_{a_2} \lambda_2}{(\lambda_1-\lambda_2)^2  }.\nn\\
\label{24-matrix}
\eea
\ep

\br \label{remarkremark}
The right hand side of~\eqref{24} and the right hand side of~\eqref{24-matrix}
coincide. However,  this does not mean the summands coincide with each other. 
\er

\setcounter{equation}{0}
\setcounter{theorem}{0}
\subsection{An algorithm for writing the DS-hierarchy}
Let~$\V$ be any gauge of DS-type, $\{X^1,\dots,X^n\}$ 
a basis of~$\V$ s.t. $\deg X^a=-m_a$ and let
$$\L^{\rm can}=\p_x+\Lambda+ q^{\rm can}(x),\qquad q^{\rm can}(x)=\sum_{a=1}^n w_a(x) X^a. $$
Denote by~$R^{\rm can}_a$, $a=1,\dots,n$ the basic resolvents of~$\L^{\rm can}$. 
The corresponding DS-hierarchy will be defined as in~\eqref{DS-qcan-w-precise}.
Although we know that the right hand side of~\eqref{DS-qcan-w-precise} 
depends only on $q^{\rm can}, q^{\rm can}_x,\cdots$, 
the second term of the right hand side of~\eqref{DS-qcan-w-precise} 
contains evolution of general components in~$\fn$.

So the following question is under consideration:

\textit{For any given gauge $\mathcal{V}$, can we write down the DS-hierarchy for $q^{\rm can}$ using only the information of~$R_a^{\rm can}$?}

\vspace{2mm}

Let us give a positive answer to this question by using our definition of tau-function.
\begin{itemize}
\item[1.]  Compute the basic resolvents $R_a^{\rm can}$, $a=1,\dots,n$.
\item[2.]  Calculate the Miura transformation $w_a \mapsto r_a$ 
 from eq.~\eqref{gen-ham}. Recall that the normal coordinates are defined by $r_a:=\langle\langle \tau_{a,0}\tau_{1,0} \rangle\rangle^{DS}.$
\item[3.]  
Calculate $\langle\langle \tau_{b,k} \tau_{a,0}\rangle\rangle^{DS}$ from eqs.~\eqref{gen-simple-2point}.
Note that the DS-flows for the normal coordinates~$r_a$ are
\be\label{normal-ru}
\frac{\p r_a}{\p T^{b}_k} = - \p_x \, \langle\langle \tau_{b,k} \tau_{a,0}\rangle\rangle^{DS},\qquad a,b=1,\dots,n,\,k\geq 0.
\ee
The right hand sides of eqs.~\eqref{normal-ru} are differential polynomials in~$w$. 
Substituting $w_a \mapsto r_a$
  in the right hand sides of eqs.~\eqref{normal-ru} we obtain the DS hierarchy for~$r_a$.
\item[4.] Substitute the inverse Miura transformation to the 
DS hierarchy for~$r_a$ we obtain the DS hierarchy.
\end{itemize} 

\setcounter{equation}{0}
\setcounter{theorem}{0}
\section{Computational aspect of resolvents} \label{essential}
\setcounter{equation}{0}
\setcounter{theorem}{0}
\subsection{The lowest weight gauge}
Recall that there is a particular choice of a gauge of DS-type \cite{BFRFW}, called
the \textit{lowest weight gauge}. Let us review its construction. Write the Weyl co-vector 
as $\rc =\sum_{i=1}^n x_i H_i$, $x_i\in \CC$ and define 
\be \label{def-I-}
I_-=2\sum_{i=1}^n x_i F_i.
\ee
Then $I_+,I_-,\rc $ generate an $sl_2(\mathbb{C})$ Lie subalgebra of $\g$:
\be
\label{SL2princ}
[\rc , I_+]= I_+,\quad [\rc ,I_-]=-I_-,\quad [I_+,I_-]=2\rc .
\ee
According to \cite{Kostant,BFRFW}  there exist elements $\gamma^1,\dots,\gamma^n\in \g$ such that
$$
{\rm Ker} \, {\rm ad}_{I_-} = {\rm Span}_{\mathbb{C}} \le\{\gamma^1,\dots,\gamma^n\ri\},\quad  [\rc ,\gamma^i]=-m_i \, \gamma^i.
$$
Since $\gamma^n\in \C \Ftheta$ we could and will normalize it to be 
\be\label{norm-gamman}\gamma^n = \Ftheta.\ee 
The subspace ${\rm Ker} \, {\rm ad}_{I_-} \subset \fb$ is a gauge of DS-type, 
which is called the lowest weight gauge.
Denote by
$$\L^{\rm can}=\p_x+\Lambda+q^{\rm can}(x)$$
the gauge fixed Lax operator associated to ${\rm Ker} \, {\rm ad}_{I_-}$, where 
$q^{\rm can}(x) := \sum_{a=1}^n u_a(x) \gamma^a.$
\bd
The functions $u_a$, $a=1,\dots,n$ are called the {\bf lowest weight coordinates}.
\ed

\setcounter{equation}{0}
\setcounter{theorem}{0}
\subsection{Extended principal gradation}

\bd
Define the {\bf extended principal degree} by the following degree assignments
\bea
&& \deg^e \p_x=1,\quad \deg^e  \lambda=h,\\
&& \deg^e u_i= m_i+1,\\
&& \deg^e E_i=1,\quad \deg^e F_i=-1, \quad i=1,\dots, n. 
\eea
\ed
\noindent
It is easy to see that, 
if $a\in L(\g)^j$  then $\deg^e a=\deg a=j.$ Namely, the extended principal degree coincides with the principal degree for any loop algebra element. In particular, 
\be
\deg^e \gamma^i=-m_i,\quad \deg^e {\rm ad}_{I_+}^j \gamma^i=-m_i+j,\qquad j=0,\dots, 2m_i.
\ee
\bl \label{degL}
For the gauge-fixed Lax operator $\L^{\rm can}$, we have $\deg^e\L^{\rm can}=1$.
\el

Let $(U^{\rm can},H^{\rm can})$ be the unique pair associated to $\L^{\rm can}$, and $R^{\rm can}_a$ the basic resolvents. 

\bl \label{deg-U-H-R}
The following formulae hold true
\be
\deg^e U^{\rm can}=0,\quad \deg^e H^{\rm can}=1,\quad  \deg^e R_a^{\rm can}= m_a,\qquad a=1,\dots,n.
\ee
\el
\begin{proof}
By using the recursion procedure \eqref{recur} and by the mathematical induction.
\end{proof}

\bc  \label{deg-F}
The $N$-point ($N\geq 2$) generating series of correlation functions  
$F_{a_1,\dots,a_N}(\lambda_1,\dots,\lambda_N;\bdT)$ are homogenous of degree $-Nh+\sum_{\ell=1}^N m_{a_{\ell}}$ w.r.t. the extended principal gradation.
\ec

\setcounter{equation}{0}
\setcounter{theorem}{0}
\subsection{Essential series of the Drinfeld--Sokolov hierarchy}

Recall that the simple Lie algebra $\g$ admits the lowest weight decomposition \cite{BFRFW}
$$\g=\bigoplus_{a=1}^n \, \mathfrak L^a,\qquad \mathfrak L^a=
{\rm Span}_{\CC} \le\{\gamma^a,\,\ad_{I_+}\, \gamma^a,\dots, \ad_{I_+}^{2m_a}\,\gamma^a\ri\},$$
where each~$\mathfrak L^a$ is an $sl_2(\CC)$-module w.r.t. the $sl_2(\CC)$ Lie subalgebra generated by
$I_+,I_-,2\rho^\vee$, called a lowest weight module. It is then clear that any 
$\mathfrak g$-valued function $R(\lambda)$ can be uniquely written as
$$
R(\lambda)=\sum_{a=1}^n \sum_{m=0}^{2m_a} \, K_{am}(\lambda) \, \ad_{I_+}^{m} \gamma^a.
$$ 

\bt \label{resolvent-reduction}
Let $\L^{\rm can}=\p_x+\Lambda+q^{\rm can}=\p_x+\Lambda+\sum_{a=1}^n \, u_a \,\gamma^a$ be a Lax operator associated to the lowest weight gauge.   
Let $R^{\rm can}\in \mathcal{A}^u\otimes \g ((\lambda^{-1}))$ be any resolvent of $\L^{\rm can}.$ Write
\be \label{dec-Rcan}
R^{\rm can}=\sum_{i=1}^n \, \mathcal{R}_i \, \ad_{I_+}^{2m_i} \gamma^i+ \sum_{i=1}^n \sum_{m=0}^{2m_i-1} \, K_{im} \, \ad_{I_+}^{m} \gamma^i.
\ee We have
1) $\forall\, i\in\{1,\dots,n\},\,m\in\{0,1,\dots,2m_i-1\}$, 
$K_{im}$ has the following expression
$$
K_{im}= \sum_{j=1}^n \sum_{\ell=0}^{2m_i-m} \left(s_{i,\ell,0}^j +\lambda \, s_{i,\ell,1}^j\right) \, \p_x^\ell \left(\mathcal{R}_j\right),
$$
where the coefficients $s_{i,\ell,0}^j, s_{i,\ell,1}^j$ belong to $\mathcal{A}^u$, and they do not depend on the choice of the resolvent. 

2) The ODE $[\L^{\rm can},R^{\rm can}]=0$ is equivalent to $n$ scalar linear ODEs for $\ERR_1,\dots,\ERR_n$.

3) The following formulae hold true for the degrees of the coefficients \eqref{dec-Rcan} of the basic resolvents
\be\label{prin-R-K}
\deg^e \mathcal{R}_{a; \,i}=m_a-m_i,\quad \deg^e K_{a; \, im}=m_a+m_i-m, \quad  i,a=1,\dots,n; \,m=0,\dots,2m_i-1.
\ee
\et

\noindent{\textit{Proof}} of Theorem~\ref{resolvent-reduction}. ~ 
Write
$$
R^{\rm can}(\lambda; u; u_x,\dots) = \sum_{i=1}^n  \sum_{m=0}^{2m_i}  K_{im}(\lambda; u; u_x,\dots) \,  {\rm ad}_{I_+}^m \gamma^i,\qquad K_{i,2m_i}:=\mathcal{R}_i.
$$
Substituting the above expressions into \eqref{ODE-R} we obtain
\be
\sum_{i=1}^n \sum_{m=0}^{2m_i}  
\frac{ \p K_{im}}{\p x}  \,  {\rm ad}_{I_+}^m \gamma^i + 
\sum_{i=1}^n \sum_{m=1}^{2m_i} K_{i,m-1}  \, {\rm ad}_{I_+}^{m} \gamma^i 
+\le[\lambda\,\gamma^n+\sum_{\ell=1}^n u_\ell \gamma^\ell, \,\sum_{i=1}^n  \sum_{m=0}^{2m_i}  
K_{im}\,  {\rm ad}_{I_+}^m \gamma^i\ri]=0.\label{106}
\ee
Introduce the lowest weight structure constants $c_{\ell ijs}^m$ by
\be \label{107}
[\gamma^\ell,\,\ad_{I^+}^m \gamma^i]=\sum_{j=1}^n \sum_{s=0}^{2m_j} \, c_{\ell ijs}^m \, \ad_{I^+}^s \gamma^j,\qquad i,\,\ell=1,\dots,n,\,m=0,\dots,2m_i.
\ee
Substituting \eqref{107} into \eqref{106} we obtain 
\bea\label{108}
&& 
\sum_{i=1}^n \sum_{m=0}^{2m_i}  
\frac{ \p K_{im}}{\p x}  \,  {\rm ad}_{I_+}^m \gamma^i + 
\sum_{i=1}^n \sum_{m=1}^{2m_i} K_{i,m-1}  \, {\rm ad}_{I_+}^{m} \gamma^i \nn\\
&& 
+\sum_{\ell=1}^n \sum_{i=1}^n  \sum_{m=0}^{2m_i}\sum_{j=1}^n \sum_{s=0}^{2m_j} \, \wt{u_\ell} \, K_{im} \, c_{\ell ijs}^m \, \ad_{I^+}^s \gamma^j=0,
\eea
where $\wt{u_\ell}=u_\ell+\lambda\,\delta_{\ell,n}$.
It follows that
\be
K_{j,s-1}+\frac{ \p K_{js}}{\p x}+\sum_{\ell=1}^n \sum_{i=1}^n  \sum_{m=0}^{2m_i} \, \wt{u_\ell} \, K_{im} \, c_{\ell ijs}^m=0,\qquad  j=1,\dots,n,\, s=0,\dots,2m_j.
\ee
Here $K_{j,-1}:=0$.
Noting that the structure constant $c_{\ell ijs}^m$ are zero unless 
\be
0\leq m=m_i+m_\ell+s-m_j\leq 2m_i.
\ee
Hence we obtain
\be\label{Kjs-1}
K_{j,s-1}=-\frac{ \p K_{js}}{\p x} - \sum_{\ell, \, i=1 \atop m_i\geq |m_\ell+s-m_j|}^n 
\wt{u_\ell} K_{i,m_i+m_\ell+s-m_j}  c_{\ell ijs}^{m_i+m_\ell+s-m_j} , \quad  j=1,\dots,n,~s=0,\dots,2m_j.
\ee
Define an ordering for pairs of integers $\{(j,s) \,| \,j=1,\dots,n,\,s=0,\dots,2m_j\}:$ we say
$(j_1,s_1)>(j_2,s_2)$, if $s_1>s_2$, or $s_1=s_2$ and $j_1<j_2$.
Noting that 
$K_{i,2m_i}:=\mathcal{R}_i$ we can use~\eqref{Kjs-1} to 
solve out~$K_{j,s-1}$ in terms of~$\ERR_j$ and their $x$-derivatives 
starting from the largest pair $(j,s-1)=(n,2m_n-1)$ to the smallest pair $(j,s-1)=(n,0)$.
This proves Part~1) of the theorem.

Taking $s=0$ in~\eqref{Kjs-1} we obtain the system of ODEs for $\ERR_1,\dots,\ERR_n$, which proves Part~2).

Formulae~\eqref{prin-R-K} follow from Lemma~\ref{deg-U-H-R} 
and eq.~\eqref{dec-Rcan}, which proves Part~3).
$\hfill\Box$

\bd
We call $\ERR_{a;1},\dots, \ERR_{a;n}$ the {\it  essential series} of the DS hierarchy of the $\mathfrak g$-type.
\ed
Using the same argument as in \cite{BDY2}, the essential series $\ERR_{a;a}$ never vanishes.
\bd
We call $\ERR_{a;a}$ the fundamental series of the DS hierarchy.
\ed

\setcounter{equation}{0}
\setcounter{theorem}{0}
\section{Proof of Theorem \ref{topo-thm}} \label{main proof}
\setcounter{equation}{0}
\setcounter{theorem}{0}
\subsection{Relation between normal coordinates and lowest weight coordinates}\label{NL-Miura}
The concept of normal coordinates was introduced in \cite{DZ-norm}; see also \cite{DLYZ1}.
\bd \label{normal-def}
We call $r_a:=\langle\langle\tau_{a, 0} \tau_{1,0}\rangle\rangle^{DS}$ the normal coordinates of the DS hierarchy.
\ed
Recall that
$$
\Lambda_{m_a}(\l) =  L_{m_a} +  \l  \, K_{m_a-h} , \qquad L_{m_a} \in \g^{m_a} , ~ K_{m_a-h} \in \g^{m_a-h}.
$$
Using the commutativity between $\Lambda_{m_1},\dots, \Lambda_{m_n}$  
along with the normalization \eqref{norm-Lambda-2}
we have 
\bea 
&\& [L_{m_a},L_{m_b}]=0,\quad [K_{m_a-h},K_{m_b-h}]=0,\\
&\& [K_{m_a-h}, L_{m_b}] + [L_{m_a}, K_{m_b-h}]=0 \label{comm-KL}
\eea
and~\eqref{LK}.
Note that $L_{m_1}=I_+$,  we have in particular 
\be
[I_+, L_{m_a}] =0, \quad \forall \, a=1,\dots, n.
\label{hwv}
\ee
Therefore, the elements~$L_{m_a}$ are the highest weight vectors of the 
lowest weight module~$\L^a$, i.e., 
$$L_{m_a} = {\rm const} \,  \ad _{I_+}^{2m_a} \gamma^a,\quad {\rm const}\neq 0.$$ 
 
\begin{lemma} \label{norm-g-lem}
The lowest weight vectors $\gamma^a$  can be normalized such that
\be
 \label{normgamma}
 (\gamma^a \,|\, L_{m_a})=1.
\ee
\end{lemma}
\begin{proof}
We know that different irreducible representations of $sl_2(\CC)$ are orthogonal w.r.t. to $(\cdot|\cdot)$ and, hence, the nondegeneracy of $(\cdot|\cdot)$ 
implies the nondegeneracy of its restriction to each irreducible representation. 
Note that 
$$\Bigl(\gamma^a \mid  \ad_{I_-}^k L_{m_a}\Bigr)=- \Bigl(I_- \mid [\gamma^a, \ad_{I_-}^{k-1} L_{m_a}]\Bigr)=0,\quad \forall \,k\in \{1,\cdots, 2m_a\}.$$ 
So $(\gamma^a  \mid L_{m_a})\neq 0$ since otherwise we obtain a contradiction with the nondegeneracy of $(\cdot|\cdot)$.
Hence for $a=1,\cdots,n-1$,
we can normalize $\gamma^a$ such that
$(\gamma^a \,|\, L_{m_a})=1$.
Particular consideration must be addressed for $\gamma^n$, since we have already
 defined $\gamma^n=E_{-\theta}$.
Taking in~\eqref{LK} $a=n$, $b=1$ we obtain
$$(L_{m_n}\,|\,K_{m_1-h})=1 ~ \Rightarrow ~ (L_{m_n}\,|\,E_{-\theta})=1,$$
which finishes the proof.
\end{proof}

From now on we fix a choice of $\gamma^1,\cdots,\gamma^n$ satisfying \eqref{normgamma}. 
Then Lemmas~\ref{van-Cartan}, \ref{norm-g-lem} imply
\be\label{orthonormal}
(\gamma^a\,| \, L_{m_b}) = \delta^a_b.
\ee
Here it should be noted that for the case of $D_n$ with~$n$ even 
with appearance of an equal pair of exponents $m_{n/2}=m_{n/2+1}$, 
 eq.~\eqref{orthonormal} is valid under a suitable choice of 
$\gamma^{n/2}$, $\gamma^{n/2+1}$.

According to Corollary~\ref{deg-F} and Theorem~\ref{main thm}, 
$\langle\langle\tau_{a, k} \tau_{1,0}\rangle\rangle$ are differential polynomials in~$u$, 
homogeneous of degree $$m_a+1+  k h$$
w.r.t. to~$\deg^e$. 
In particular, we have
$$
\deg^e r_a = m_a+1.
$$
We arrive at the following lemma.
\bl \label{Miura-u-w}
There exists a Miura transformation $u\rightarrow r$  of the form
\be \label{Miura-transformation1}
r_a=  c_a u_a + P_a \le[u_1,\dots,u_{a-1}\ri]
\ee
for some non-zero constants $c_a$, where $P_a$ are differential polynomials in $u_1,\dots,u_{a-1}$ satisfying
\be
\deg^e P_a \le[u_1,\dots,u_{a-1}\ri]=m_a+1.
\ee 
\el
\br
For $D_{n}$ with $n$ even, Lemma~\ref{Miura-u-w} is valid under a suitable 
choice of $\gamma^{n/2}$, $\gamma^{n/2+1}$.
\er
\br
The inverse Miura transformation has the form
\be
u_a = c_a^{-1}  r_a + \wt P_a \le[ r_1,\dots, r_{a-1}\ri],
\ee
thanks to the triangular nature of the transformation~\eqref{Miura-transformation1}.
\er
\bl \label{lemmaconst}
The constants $c_a$ in Lemma \ref{Miura-u-w} have the following explicit expressions
\be
c_a = - \frac {m_a}{h}.
\ee
\el
\noindent {\bf Proof.}
Fix $a\in \{ 1,\dots, n\}$. We are to compute $r_a|_{u_1,\dots,u_{a-1}\equiv 0}$. 
Assume $u_1\equiv 0$, $\cdots$, $u_{a-1}\equiv 0$.
Looking at equation \eqref{recurrUH} for the pair $(U,H)$ we obtain
$$U^{[-1]}=\cdots= U^{[-m_a]}=0=\HH^{[-1]}=\cdots=\HH^{[1-m_a]}.$$ 
The first nontrivial equation arises from the component of principal degree $-m_a$ in \eqref{recurrUH}:
\be
\HH^{[-m_a]}  + \le[  U^{[-m_a-1]}, \Lambda\ri] = u_a \gamma^a\qquad (\mbox{no summation over~} a).
\label{227}
\ee
Let us decompose the elements $\HH^{[-m_a]} , U^{[-m_a-1]}$ as follows
\bea
&& \HH^{[-m_a]}  = \frac   {g_a(x)} \l \, \Lambda_{h-m_a}= g_a(x) \,  K_{-m_a} + \frac {g_a(x)} \l \, L_{h-m_a},\quad a=1,\dots,n, \nn \\
&& U^{[-m_a-1]} = \frac 1 \l Y_{h-m_a-1} + W_{-m_a-1},\qquad a=1,\dots, n-1, \nn \\
&& U^{[-m_n-1]} = \frac 1 \l Y_{0}. \nn
\eea
Substituting these expressions in \eqref{227} and comparing  the coefficients of powers of $\l$ we obtain
\bea
& \l^{-1}: & g_a(x)\, L_{h-m_a} + [Y_{h-m_a-1} , I_+]  =0, \label{mmeq-1} \\
& \l^{0}: & g_a(x)\,  K_{-m_a}  + [Y_{h-m_a -1}, \Ftheta]  + [W_{-m_a-1}, I_+]= u_a \gamma^a, \label{mmeq0}\\
& \l^{1}: & [W_{-m_a-1}, \Ftheta]=0 \hbox {\qquad(automatic!).} \label{mmeq1}
\eea
Since~$L_{h-m_a}$ is the highest weight vector of the irreducible $sl_2(\mathbb{C})$-module~$\L^{n+1-a}$,
the solution to eq.~\eqref{mmeq-1} is
$$
Y_{h-m_a -1} =\frac {g_a(x)} {2(h-m_a) }  [I_-, L_{h-m_a}] + f(x) L_{h-m_a -1}
$$
for some function~$f(x)$ which is a differential polynomial in~$u$. 
Here $L_{h-m_a-1}$ is defined to be $0$ if $h-m_a-1$ is not an exponent. 
We thus have
\bea
 [Y_{h-m_a-1} , \Ftheta] 
&=& \frac {g_a(x)}{2(h-m_a)} \le[I_-,[L_{h-m_a }, \Ftheta]\ri] +f(x)\, [L_{h-m_a -1}, \Ftheta] \nn\\
&\ds \mathop{=}^{\eqref{comm-KL}}&
 \frac {g_a(x)} {2(h-m_a)}[I_-,[I_+, K_{-m_a}]]+ f(x) \, [L_{h-m_a -1}, \Ftheta].
 \label{QF}
\eea
Plugging \eqref{QF} into \eqref{mmeq0} we find
\bea
 g_a(x) K_{-m_a}  + \frac {g_a(x)} {2(h-m_a) }[I_-,[I_+, K_{-m_a}]] + [W_{-m_a-1}, I_+]
 +f(x)\le[L_{h - m_a -1}, \Ftheta\ri]= u_a \gamma^a. \nn
\eea
Employing the Jacobi identity we obtain
$$
 g_a(x) \, \frac {h}{h-m_a } \, K_{-m_a}+ \le [I_+ \, , \,\frac {g_a(x)} {2(h-m_a) }[I_-, K_{-m_a}] -W_{-m_a-1} \ri] 
 +  f(x) \le[L_{h - m_a -1}, \Ftheta\ri]= u_a \gamma^a.
$$
Taking the inner products of both sides of the above equation with $L_{m_a}$
we have
\be
\Biggl(L_{m_a} \bigg |  \, \frac {h \, g_a(x) }{h-m_a } \, K_{-m_a}+ \le [I_+ \,,\,\frac {g_a(x) \, [I_-, K_{-m_a}]}{2(h-m_a) } 
-W_{-m_a-1} \ri] + f(x) \le [L_{h - m_a -1}, \Ftheta \ri] \Biggr) = u_a \le(L_{m_a}  | \gamma^a\ri). \nn
\ee
Noticing that~$L_{m_a}$ is a highest weight vector of the $sl_2(\CC)$-module~$\L^a$, i.e.,
$$
[L_{m_a},I_+]=0,\qquad [L_{m_a}, L_{h - m_a -1}]=0,
$$
and using \eqref{LK}, \eqref{normgamma} we obtain
$$
g_a(x)  
= \frac {h-m_a}{h \le(L_{m_a}  | K_{-m_a} \ri) } \, {\le(L_{m_a} | \gamma^a\ri)} \, u_a(x) =   
\frac {1}{h}  u_a(x) .
$$

Using Definition~\ref{normal-def} and eq.~\eqref{gen-ham} we have
$$
- r_a =  \res{\l=\infty} \Big(e^{U} \Lambda_{m_a} e^{-U} \Big| \Ftheta   \Big) 
= \res{\l=\infty}\biggl( \Lambda_{m_a}(\l) \Big| 
\Ftheta - \le[U(\l),\Ftheta\ri] + \frac  1 2 \le[U(\l),\le[U(\l),\Ftheta\ri]\ri] + \cdots\biggr).
$$
The only possible contribution to the residue comes from the terms of principal degree $-h-m_a$ and the first one in the series is easily seen to be residueless 
$$\res{\l=\infty} \le(\Lambda_{m_a}(\l)|\Ftheta\ri) d\lambda=0.$$
Note that we have already shown that $U$ 
has the form
\bea
U= U^{[-m_a-1]}  + \sum_{j\leq -m_a-2} U^{[j]}. \nn
\eea
Therefore only the very next term $- \le(\Lambda_{m_a}(\l)  \mid [ U(\l), \Ftheta]\ri)$ can contribute to the residue. Thus 
\be
r_a = \res{\l=\infty}  \le(  \Lambda_{m_a}(\l) \mid [U(\l),\Ftheta]  \ri) = \res{\l=\infty} \le( \Lambda_{m_a}(\l) \mid  [U^{[-m_a-1]}(\l),\Ftheta] \ri). \label{230}
\ee
Now 
substituting 
\be
\Lambda_{m_a}(\l) =  \l  \, K_{m_a-h} +  L_{m_a},\quad U^{[-m_a-1]} = \frac 1 \l \, Y_{h-m_a-1} + W_{-m_a-1}\nn
\ee
in \eqref{230}
we obtain
\bea
- r_a(x) &=&  \le( L_{m_a} \bigg|  [Y_{h-m_a -1},\Ftheta] \ri) 
 =  \le( L_{m_a} \bigg|  \le [\frac {g_a(x)} {2(h-m_a)}  [I_-, L_{h-m_a}] + f(x)\, L_{h - m_a -1},\Ftheta \ri] \ri) \nn\\
&=&
  \frac {g_a(x)}{2(h-m_a) }\le( \, L_{m_a} \,| \,  \le [  [\Ftheta \,,\, L_{h-m_a}]\,,\,I_- \ri] \,\ri) 
 \nn \\  
 & =&
  \frac {g_a(x)}{2(h-m_a) }\le( \, L_{m_a} \, | \,  \le [  [K_{-m_a}\,,\,I_+] \, , \,I_- \ri] \, \ri) 
 =
  \frac {g_a(x)}{2(h-m_a)}\le( \, \le[ I_+\,,\, [I_-\,,\, L_{m_a}]\ri] \,|\,   K_{-m_a} \, \ri) \nn \\
&=&
  g_a(x)\, \frac {m_a}{h-m_a} \le( \,L_{m_a} \,|\,   K_{-m_a} \,\ri)= \frac {m_a}{h}\,u_a(x).\nn
 \eea
The lemma is proved. \QED

\br
For the particular $A_n$ case, a similar lemma on relations between normal coordinates and Wronskian-gauge coordinates was obtained for example in~\cite{BY} (see Lemma~3.1 therein). However, except the $A_1$ case,
 the Wronskian-gauge coordinates are {\it not} the lowest weight coordinates.
\er

\setcounter{equation}{0}
\setcounter{theorem}{0}
\subsection{Partition function and topological ODE} \label{topo}
Recall that the partition function~$Z$ of the DS hierarchy of $\g$-type is a particular tau-function specified (up to a constant factor) by the string 
equation \eqref{string}.  The compatibility between the string equation and the DS hierarchy follows from the fact that 
the flow $\p_{s_{-1}}$ defined via
$$\p_{s_{-1}} \tau := \sum_{a=1}^n \sum_{k\geq 0} t^a_{k+1} \frac{\p \tau}{\p t^a_k}+\frac12 \sum_{a,b=1}^n \, \eta_{ab} \,  t^{a}_0  \, t^b_0 \tau - \frac{\p \tau}{\p t^1_0}$$
gives rise to an additional infinitesimal symmetry of the DS hierarchy.
 
The function $u=u(\bdT)=u(\bdt)$ associated to $Z(\bdt)$ 
is called the topological solution to the lowest-weight-gauge DS hierarchy, and
$r=r(\bdt)=r(\bdT)$ 
the topological solution in normal coordinates.

\bl \label{w-ini}
The normal coordinates associated to the partition function $Z$ satisfy
\be
r_{a}(\bdt)\big|_{t^{b}_k=\delta^b_1 \delta_{k,0} t^1_0} = 
- \delta_{a,n}  \frac{h-1}{h \kappa} t^1_0,\qquad \kappa:=\sqrt{-h}^{-h}.
\ee
\el
\begin{proof}
By applying the $t^a_0$-derivative on both sides of  eq. \eqref{string} we have
$$
\frac{\p^2 \log Z}{\p t^1_0 \, \p t^a_0}\Big|_{t^b_k=\delta^b_1 \delta_{k,0} t^1_0}=\delta_{a,n}  t^1_0.
$$
Hence from \eqref{norm-tT} we obtain
$$
\frac{\p^2 \log Z}{\p T^1_0 \, \p T^a_0}\Big|_{t^b_k=\delta^b_1 \delta_{k,0}  t^1_0}= 
- \delta_{a,n}  \frac{h-1}{h} \sqrt{-h}^h  t^1_0.
$$
The lemma is proved.
\end{proof}

\bl \label{u-ini} 
The topological solution to the lowest-weight-gauge DS hierarchy of $\g$-type satisfies
\be
u_a(\bdt)\big|_{t^b_k=\delta^b_1 \delta_{k,0}  t^1_0}=  \frac{\delta_{a,n}} {\kappa}  t^1_0.
\ee
\el
\begin{proof}
By applying Lemma \ref{Miura-u-w}, Lemma \ref{lemmaconst} and Lemma \ref{w-ini}.
\end{proof}

\paragraph{Topological ODE of $\g$-type.}  Let $u=u(\bdT)=u(\bdt)$ be the topological solution 
to the lowest-weight-gauge DS hierarchy, and $R^{\rm can}_{a}(\l;\bdt)$ be the basic resolvents of $\L^{\rm can}$ (see Definition \ref{resol-defi}). Note that 
$$t^1_0=-T^1_0=x.$$
Define $$R_{a}^{\rm can}(\lambda, x) =   \lambda^{-\frac{m_a} h} \, R^{\rm can}_a(\l;\bdt) \big|_{t^{b}_k= x\,\delta^b_1 \, \delta_{k,0} }, \qquad a=1,\dots,n\,.$$ 
Clearly, $R^{\rm can}_{a}(\l,x)$ is the unique solution to \eqref{ODE-R}--\eqref{convention-R} with $\L$ replaced by 
$\L^{\rm can}=\p_x + \Lambda + \frac{x}{\kappa}E_{-\theta}$.

\bl  [Key Lemma] \label{key lemma} The following formulae hold true
\be \label{x-l}
\p_x\,\bigl(R_{a}^{\rm can}\bigr)= \frac1 \kappa  \, \p_\lambda\, \bigl(R_{a}^{\rm can}\bigr),\qquad a=1,\dots,n.
\ee
\el
\begin{proof} 
For each $a\in\{1,\dots,n\}$, let $M_a^*(\lambda)$ be the unique solution to the topological ODE \eqref{topoode} satisfying
$$
M_a^*(\l)=\lambda^{-\frac{m_a}h} \left[ \Lambda_{m_a}(\l)+ \mbox{lower degree terms w.r.t. } \deg\right].
$$ 
See in \cite{BDY2} for the proof of existence and uniqueness of $M_a^*(\lambda)$.
Define $R_a^{\rm can,*}(\lambda,x)= \lambda^{\frac{m_a}h} M_a^*(\lambda+\frac x \kappa)$.
Then $R_a^{\rm can,*}$ satisfies equations \eqref{ODE-R}--\eqref{convention-R} with $\L=\p_x + \Lambda + \frac{x}{\kappa}\gamma^n$.  
Hence the uniqueness statement of Proposition \ref{boundary} implies that $R_a^{\rm can}(\lambda,x)\equiv R_a^{\rm can,*}(\lambda,x)$, $a=1,\dots,n$.
The lemma is proved.
\end{proof}

\noindent {\it Proof}~of~Theorem~\ref{topo-thm}. \quad Note that $M_a(\lambda):=R^{\rm can}_a(\lambda;x=0).$
So from the above proof of Lemma \ref{key lemma} we already see that $M_a$ satisfies the 
topological ODE~\eqref{topoode}.
The theorem is proved.
\hfill$\Box$

\noindent{\it Proof}~of~Theorem~\ref{N-DS-p-real}.\quad 
By using Theorem-ADE, Theorem-BCFG, Theorem~\ref{main thm}, 
and by using Theorem~\ref{topo-thm} we obtain 
\bea
&& 
 (\kappa\,\sqrt{-h})^N \sum_{g,k_1,\cdots,k_N\geq 0} (-1)^{k_1+\dots+k_N}\prod_{\ell=1}^N 
\frac{\left( \frac{m_{i_\ell}}{h}\right)_{k_\ell+1}}{\left( \kappa \, \widetilde \lambda_\ell\right)^{\frac{m_{i_\ell}}{h}+k_\ell+1}} \langle\tau_{i_1 k_1} \dots \tau_{i_N k_N}\rangle^{\g}_g \nn\\
&=& - \frac 1 {2 N\, h^\vee} \sum_{s \in S_N} 
\frac {B\le( \widetilde M_{i_{s_1}} \le( \widetilde \lambda_{s_1}\ri),\dots, \widetilde M_{i_{s_N}}\le(\widetilde \lambda_{s_N}\ri) \ri) }{\prod _{j=1}^{N} \le(\widetilde \lambda_{s_j}- \tilde \lambda_{s_{j+1}}\ri)}\nn\\
&\&  \qquad\qquad\qquad\qquad\qquad -\delta_{N2} \, \eta_{i_1 i_2}  \frac {\widetilde \lambda_1^{-\frac {m_{i_1}}h } \widetilde \lambda_2^{- \frac {m_{i_2}}h} \le(m_{i_1} \, \widetilde \lambda_1 + m_{i_2} \, \widetilde \lambda_2\ri)}
{\Bigl(\widetilde \lambda_1 - \widetilde \lambda_2\Bigr)^2}, \quad N\geq 2, \label{previous-version-N}
\eea
where $\widetilde M_a = \widetilde M_a\bigl(\tilde \lambda\bigr)$, $a=1, \dots, n$ are the unique solutions to
\bea
&& \frac{{\rm d} \widetilde M}{{\rm d} \tilde \lambda}= \kappa \, \le[\widetilde{M},\Lambda\bigl(\tilde\lambda\bigr)\ri],\quad \kappa=\le(\sqrt{-h}\ri)^{-h},\nn\\
&& \widetilde M_a\bigl(\tilde \lambda\bigr) = \tilde \lambda^{-\frac{m_a}h} \Bigl[ \Lambda_{m_a}
\bigl(\tilde \lambda\bigr) + \mbox{lower degree terms w.r.t.~} \deg\Bigr].\nn
\eea
Now consider the following conjugation of $\widetilde M_a$ together with a rescaling in $\widetilde\lambda:$
\bea
&& M_a(\lambda) = \sigma^{\rho^\vee} \, \widetilde M_a \bigl(\tilde \lambda\bigr) \, \sigma^{-\rho^\vee}, \nn\\
&& \lambda = \sigma^{-h} \,  \tilde \lambda, \nn
\eea
where $\sigma:=\kappa^{-\frac1{h+1}}$. It is straightforward to check that
\bea
&&  \frac{{\rm d} M}{{\rm d} \lambda} =   [ M,\Lambda(\lambda)],\nn\\
&&  M_a(\lambda) =  \lambda^{-\frac{m_a}h} 
\Bigl[ \Lambda_{m_a}(\lambda)+ \mbox{lower degree terms w.r.t.~} \deg\Bigr].\nn
\eea
Combining with \eqref{previous-version-N}, this proves the validity of the formula \eqref{24-FJR}.
To prove formula \eqref{N=1-FJR}, one further needs to observe the following identity obtained from the string equation \eqref{string}
$$
\langle\tau_{a, k+1}\tau_{1,0} \rangle^{FJRW-\g}=\langle\tau_{ak}\rangle^{FJRW-\g}, \quad a=1,\dots,n, ~ k\geq 0.
$$
The rest of proving \eqref{N=1-FJR} follows from the identity~\eqref{gen-ham} 
and the above conjugation of~$\widetilde M_a$ with the rescaling in~$\tilde \lambda$.
\hfill$\Box$

\noindent \textit{Proof} of Theorem~\ref{r-spin-thm}. 
The theorem is a particular case of Theorem~\ref{N-DS-p-real} (cf. Remark~\ref{remarkremark}) with the particular realization of~$A_n$ Lie algebra being consistent with normalization of flows suggested by Witten~\cite{W2}.
\QED

\begin{example}[Rationality of Witten's $r$-spin intersection numbers]
It is known that Witten's $r$-spin intersection numbers are 
rational numbers. Let us verify the rationality through \eqref{easy-1} and \eqref{r-spin-N-new}. Indeed, our definition of $N$-point $r$-spin correlators reads
\bea
F^{r-spin}_{a_1,\dots,a_N}(\lambda_1,\dots,\lambda_N)
&=& \left(\kappa^{\frac1{r+1}}\,\sqrt{-r}\right)^N \sum_{k_1,\dots,k_N\geq 0}   \prod_{\ell=1}^N \frac{ (-1)^{k_\ell}   \left(\frac{a_\ell}{r}\right)_{k_\ell+1}}{ (\kappa^{\frac1{r+1}} \,  \lambda_\ell)^{\frac{{a_\ell}}r+k_\ell+1}} \langle\tau_{a_1 k_1} \dots \tau_{a_N k_N}\rangle^{r{\rm-spin}} \nn \\ 
&=& \sum_{g\geq 0} (-r)^{g-1+N}   \sum_{k_1,\dots,k_N\geq 0}   \prod_{\ell=1}^N \frac{ (-1)^{k_\ell}   \left(\frac{a_\ell}{r}\right)_{k_\ell+1}}{ \lambda_\ell^{\frac{{a_\ell}}r+k_\ell+1}} 
\langle\tau_{a_1 k_1} \dots \tau_{a_N k_N}\rangle^{r{\rm-spin}}_g,  \nn 
\eea
where we have used $\kappa=\sqrt{-r}^{-r}$ and the dimension-degree matching~\eqref{dim-degree}. 
Clearly, all the coefficients are rational. On the other hand, 
the right hand side of~\eqref{easy-1} or of~\eqref{r-spin-N-new} 
belongs to $\mathbb{Q}[[\lambda_1^{-1/r}, \dots,\lambda_N^{-1/r}]]$ as our regular solutions~$M_a(\lambda)$, $a=1,\dots,n$ to the topological ODEs of $sl_n(\mathbb{C})$-type \eqref{topo-An}
are of rational coefficients. 
The rationality of $r$-spin correlators is verified. 
\end{example}

\appendix

\setcounter{equation}{0}
\setcounter{theorem}{0}
\section{$3$-spin} \label{3spinM}
The matrices $M_i(\lambda),\,i=1,2$ for the Witten's $3$-spin invariants have the following explicit expressions. 
Denote $M_i(\lambda)=(M_i(\lambda)^a_b)_{a,b=1,\dots,3}.$ Then we have

\begin{center}
\resizebox{0.7\textwidth}{!}{
\begin{minipage}{1\textwidth}
$$
(-M_1)^1_1 = \sum_{g\geq 0} \frac{(-1)^g \Gamma(8g+\frac43)}{12^{3g} \, g! \, \Gamma(g+\frac13)}  \lambda^{-\frac{24g+4}3} -\frac{1}{72}  \sum_{g\geq 0} \frac{ (-1)^g \Gamma(8g+\frac{16}3)}{12^{3g} \, g! \, \Gamma(g+\frac43)}  \lambda^{-\frac{24g+16}3}
$$
$$
(-M_1)^1_2= - \sum_{g\geq 0} \frac{(-1)^g \Gamma(8g+\frac{1}3)}{12^{3g} \, g! \, \Gamma(g+\frac1 3)}  \lambda^{-\frac{24g+1}3} + \frac1 {24} \sum_{g\geq 0} \frac{(-1)^g \Gamma(8g+\frac{13}3)}{12^{3g} \, g! \, \Gamma(g+\frac4 3)}  \lambda^{-\frac{24g+13}3}
$$
$$
(-M_1)^1_3= - \frac1{12}  \sum_{g\geq 0} \frac{(-1)^g \Gamma(8g+\frac{10}3)}{12^{3g} \,g! \,\Gamma(g+\frac43)}  \lambda^{-\frac{24g+10}3}
$$
$$
(-M_1)^2_1= \sum_{g\geq 0} \frac{(-1)^g \Gamma(8g+\frac{7}3)}{12^{3g} \, g! \,  \Gamma(g+\frac13)}  \lambda^{-\frac{24g+7}3} - \frac1{12} \sum_{g\geq 0} \frac{(-1)^g \Gamma(8g+\frac{10}3)}{12^{3g} \, g!  \, \Gamma(g+\frac43)}  \lambda^{-\frac{24g+7}3} - \frac{1}{72} \sum_{g\geq 0} \frac{(-1)^g \Gamma(8g+\frac{19}3)}{12^{3g} \, g! \, \Gamma(g+\frac4 3)}  \lambda^{-\frac{24g+19}3}
$$
$$
(-M_1)^2_2=  \frac{1}{36}  \sum_{g\geq 0} \frac{ (-1)^g \Gamma(8g+\frac{16}3)}{12^{3g} \, g! \, \Gamma(g+\frac43)}  \lambda^{-\frac{24g+16}3}
$$
$$
(-M_1)^2_3= - \sum_{g\geq 0} \frac{(-1)^g \Gamma(8g+\frac{1}3)}{12^{3g} \, g! \, \Gamma(g+\frac1 3)}  \lambda^{-\frac{24g+1}3} - \frac1 {24} \sum_{g\geq 0} \frac{(-1)^g \Gamma(8g+\frac{13}3)}{12^{3g} \, g! \, \Gamma(g+\frac4 3)}  \lambda^{-\frac{24g+13}3}
$$
$$
(-M_1)^3_1= - \frac{1}{72}\sum_{g\geq 0} \frac{(-1)^g \Gamma(8g+\frac{22}3)}{12^{3g} \, g! \, \Gamma(g+\frac4 3)}  \lambda^{-\frac{24g+22}3} - \sum_{g\geq 0} \frac{ (-1)^g \Gamma(8g+\frac{1}3)}{12^{3g} \, g! \, \Gamma(g+\frac1 3)}  \lambda^{-\frac{24g-2}3}
$$ 
$$
(-M_1)^3_2= \sum_{g\geq 0} \frac{(-1)^g \Gamma(8g+\frac{7}3)}{12^{3g} \, g! \,  \Gamma(g+\frac13)}  \lambda^{-\frac{24g+7}3} - \frac1{12} \sum_{g\geq 0} \frac{(-1)^g \Gamma(8g+\frac{10}3)}{12^{3g} \, g!  \, \Gamma(g+\frac43)}  \lambda^{-\frac{24g+7}3} + \frac{1}{72} \sum_{g\geq 0} \frac{(-1)^g \Gamma(8g+\frac{19}3)}{12^{3g} \, g! \, \Gamma(g+\frac4 3)}  \lambda^{-\frac{24g+19}3}
$$
$$
(-M_1)^3_3= - \sum_{g\geq 0} \frac{(-1)^g \Gamma(8g+\frac43)}{12^{3g} \, g! \, \Gamma(g+\frac13)}  \lambda^{-\frac{24g+4}3} -\frac{1}{72}  \sum_{g\geq 0} \frac{ (-1)^g \Gamma(8g+\frac{16}3)}{12^{3g} \, g! \, \Gamma(g+\frac43)}  \lambda^{-\frac{24g+16}3}
$$
\end{minipage}}
\end{center}
and

\begin{center}
\resizebox{0.7\textwidth}{!}{
\begin{minipage}{1\textwidth}

$$
(-M_2)^1_1= -\frac16 \sum_{g\geq 0}\frac{(-1)^g \Gamma(8g+\frac{8}3)}{12^{3g} \, g! \, \Gamma(g+\frac23)} \lambda^{-\frac{24g+8}3}  -  \frac{1}{144} \sum_{g\geq 0}\frac{(-1)^g \Gamma(8g+\frac{20}3)}{12^{3g} \, g! \, \Gamma(g+\frac53)} \lambda^{-\frac{24g+20}3}
$$
$$
(-M_2)^1_2=\frac{1}{144}\sum_{g\geq 0}\frac{ (-1)^g \Gamma(8g+\frac{17}3)}{12^{3g} \, g! \, \Gamma(g+\frac53)} \lambda^{-\frac{24g+17}3} + \frac12 \sum_{g\geq 0}\frac{(-1)^g \Gamma(8g+\frac53)}{12^{3g} \, g! \, \Gamma(g+\frac23)} \lambda^{-\frac{24g+5}3}
$$
$$
(-M_2)^1_3=-\sum_{g\geq 0} \frac{(-1)^g \Gamma(8g+\frac23)}{12^{3g} \,g! \, \Gamma(g+\frac23)} \lambda^{-\frac{24g+2}3}
$$
$$
(-M_2)^2_1= - \frac{1}{144} \sum_{g\geq 0} \frac{(-1)^g \Gamma(8g+\frac{23} 3)}{12^{3g} \, g! \, \Gamma(g+\frac53)} \lambda^{-\frac{24g+23}3} - \sum_{g\geq 0} \frac{(-1)^g \Gamma(8g+\frac23)}{12^{3g} \, g! \, \Gamma(g+\frac23)} \lambda^{-\frac{24g-1}3}  - \frac16  \sum_{g\geq 0} \frac{(-1)^g \Gamma(8g+\frac{11}3)}{12^{3g} \, g! \, \Gamma(g+\frac23)} \lambda^{-\frac{24g+11}3}
$$
$$
(-M_2)^2_2= \frac13 \sum_{g\geq 0}\frac{(-1)^g 3^{6g} \Gamma(8g+\frac{8}3)}{108^{3g} \, g! \, \Gamma(g+\frac23)} \lambda^{-\frac{24g+8}3} 
$$
$$
(-M_2)^2_3=\frac{1}{144}\sum_{g\geq 0}\frac{ (-1)^g \Gamma(8g+\frac{17}3)}{12^{3g} \, g! \, \Gamma(g+\frac53)} \lambda^{-\frac{24g+17}3} - \frac12 \sum_{g\geq 0}\frac{(-1)^g \Gamma(8g+\frac53)}{12^{3g} \, g! \, \Gamma(g+\frac23)} \lambda^{-\frac{24g+5}3}
$$
$$
(-M_{2})^3_1=-\frac{1}{6}\sum_{g\geq 0}\frac{(-1)^g \Gamma(8g+\frac{14}3)}{12^{3g} \, g! \, \Gamma(g+\frac23)} \lambda^{-\frac{24g+14}3} + \frac{1}{144}\sum_{g\geq 0}\frac{(-1)^g \Gamma(8g+\frac{17}3)}{12^{3g} \, g! \, \Gamma(g+\frac53)} \lambda^{-\frac{24g+14}3}
$$
$$
(-M_{2})^3_2= - \frac{1}{144} \sum_{g\geq 0} \frac{(-1)^g \Gamma(8g+\frac{23} 3)}{12^{3g} \, g! \, \Gamma(g+\frac53)} \lambda^{-\frac{24g+23}3} - \sum_{g\geq 0} \frac{(-1)^g \Gamma(8g+\frac23)}{12^{3g} \, g! \, \Gamma(g+\frac23)} \lambda^{-\frac{24g-1}3} + \frac16  \sum_{g\geq 0} \frac{(-1)^g  \Gamma(8g+\frac{11}3)}{12^{3g} \, g! \, \Gamma(g+\frac23)} \lambda^{-\frac{24g+11}3}
$$
$$
(-M_{2})^3_3=-\frac16 \sum_{g\geq 0}\frac{(-1)^g \Gamma(8g+\frac{8}3)}{12^{3g} \, g! \, \Gamma(g+\frac23)} \lambda^{-\frac{24g+8}3}  +  \frac{1}{144} \sum_{g\geq 0}\frac{(-1)^g \Gamma(8g+\frac{20}3)}{12^{3g} \, g! \, \Gamma(g+\frac53)} \lambda^{-\frac{24g+20}3}.
$$
\end{minipage}}
\end{center}

\newpage

\noindent Marco Bertola

\noindent SISSA, via Bonomea 265, Trieste 34136, Italy

\noindent Department of Mathematics and
Statistics, Concordia University, 1455 de Maisonneuve W., Montr\'eal, Qu\'ebec,  H3G 1M8,
Canada

\noindent Centre de recherches math\'ematiques, Universit\'e de Montr\'eal, C.~P.~6128, succ. centre ville, Montr\'eal,
Qu\'ebec, H3C 3J7, Canada

\noindent  marco.bertola@\{sissa.it, concordia.ca\}

~~~~~
~~~~~

\noindent Boris Dubrovin

\noindent SISSA, via Bonomea 265, Trieste 34136, Italy

\noindent N. N. Bogolyubov Laboratory for Geometrical Methods in Mathematical Physics, Moscow State University, Moscow 119899, Russia

\noindent {\it Deceased on March 19, 2019}


~~~~~
~~~~~

\noindent Di Yang

\noindent Max-Planck-Institut f\"ur Mathematik, Vivatsgasse 7, Bonn 53111, Germany

\noindent {\it current address}: School of Mathematical Sciences, University of Science and Technology of China, 
Hefei 230026, P.R.~China

\noindent diyang@ustc.edu.cn

\end{document}